\documentclass[journal]{IEEEtran}
\usepackage{color}
\usepackage{epsfig,graphicx}
\usepackage[skip=1ex]{caption}
\usepackage{cite}
\usepackage{subcaption}

\usepackage{amssymb,amsthm}
\usepackage{mathtools, cuted}
\usepackage{lipsum, color, colortbl}
\usepackage[utf8]{inputenc}
\usepackage{cite}
\usepackage{multirow}
\usepackage{array,booktabs}
\usepackage{multicol}
\usepackage[capitalize]{cleveref}
\usepackage{xr}
\newcolumntype{P}[1]{>{\centering\arraybackslash}p{#1}}
\definecolor{Gray}{gray}{0.9}
\newcommand{\norm}[1]{\left\lVert{#1}\right\rVert}

\newtheorem{theorem}{Theorem}

\newtheorem{lemma}{Lemma}

\newtheorem{cor}{Corollary}

\newtheorem{definition}{Definition}

\title{Information-Theoretic Bounds for Steganography in Multimedia}
\author{
	Hassan~Y.~El-Arsh,~
	Amr~Abdelaziz,~\IEEEmembership{Member,~IEEE,}
	Ahmed~Elliethy,~\IEEEmembership{Member,~IEEE,}
	Hussein~A.~Aly,~\IEEEmembership{Senior~Member,~IEEE,}
	and~T. Aaron Gulliver,~\IEEEmembership{Senior~Member,~IEEE}
	\thanks{H. Y. El-Arsh, Amr Abdelaziz, A. Elliethy, and H. A. Aly are with Dept. of Computer Engineering, Military Technical College, Cairo, Egypt.
(e-mail: hassan.yakout@ymail.com, amrashry@mtc.edu.eg, a.s.elliethy@mtc.edu.eg, haly@ieee.org).
T. Aaron Gulliver is with the Dept. of Electrical and Computer Engineering,
University of Victoria, Victoria, BC Canada. (e-mail:agullive@ece.uvic.ca)}\vspace{-20pt}
	\thanks{This paper has supplementary material available at http://ieeexplore.ieee.org, provided by the authors.}
	\thanks{Color versions of one or more of the figures in this paper are available online at http://ieeexplore.ieee.org.}
}

\newcommand{\Sender}{\mathbf{A}}
\newcommand{\Receiver}{\mathbf{B}}
\newcommand{\Eavesdropper}{\mathbf{E}}
\newcommand{\AlphabetCovSteg}{\mathbb{V}}
\newcommand{\AlphabetMsg}{\mathbb{X}}
\newcommand{\AlphabetCodedMsg}{\mathbb{M}}
\newcommand{\CoverObject}{\mathbf{c}}
\newcommand{\StegObject}{\mathbf{s}}
\newcommand{\Msg}{\mathbf{x}}
\newcommand{\CodedMsg}{\mathbf{m}}
\newcommand{\AvgProbErrorEavesdropper}{P_E}
\newcommand{\TotalProbErrorEavesdropper}{P_e}
\newcommand{\ProbErrorReceiver}{P_B}
\newcommand{\CoverDist}{P_c}
\newcommand{\StegDist}{P_s}
\newcommand{\MsgDist}{P_m}
\newcommand{\CoverElement}{c_i}
\newcommand{\StegElement}{s_i}
\newcommand{\CoverElementDist}{P_{c_i}}
\newcommand{\StegElementDist}{P_{s_i}}

\newcommand{\CoverDistc}{p_c}
\newcommand{\StegDistc}{p_s}
\newcommand{\MsgDistc}{p_m}
\newcommand{\CoverElementDistc}{p_{c_i}}
\newcommand{\StegElementDistc}{p_{s_i}}

\twocolumn

\begin{document}
\maketitle

\begin{abstract}
Steganography in multimedia aims to embed secret data into an innocent looking multimedia cover object.
This embedding introduces some distortion to the cover object and produces a corresponding stego object.
The embedding distortion is measured by a cost function that determines the detection probability of the existence of the embedded secret data.
A cost function related to the maximum embedding rate is typically employed to evaluate a steganographic system.
In addition, the distribution of multimedia sources follows the Gibbs distribution which is a complex statistical model that restricts analysis.
Thus, previous multimedia steganographic approaches either assume a relaxed distribution or presume a proposition on the
maximum embedding rate and then try to prove it is correct.
Conversely, this paper introduces an analytic approach to determining the maximum embedding rate in multimedia cover objects through a constrained optimization problem
concerning the relationship between the maximum embedding rate and the probability of detection by any steganographic detector.
The KL-divergence between the distributions for the cover and stego objects is used as the cost function as it upper bounds the performance of the optimal
steganographic detector.
An equivalence between the Gibbs and correlated-multivariate-quantized-Gaussian distributions is established to solve this optimization problem.
The solution provides an analytic form for the maximum embedding rate in terms of the WrightOmega function.
Moreover, it is proven that the maximum embedding rate is in agreement with the commonly used Square Root Law (SRL) for steganography,
but the solution presented here is more accurate.
Finally, the theoretical results obtained are verified experimentally.
\end{abstract}

\section{Introduction}

Steganography is used to embed data within innocent looking cover objects such as images, audio, video, and text \cite{steg.types},
to produce a stego object that looks similar to the original object but with hidden data embedded.
This implies some distortion of the cover object which can be exploited by a warden to detect if there is hidden data or not using steganalysis.
The goal of steganography is to hide the maximum amount of data subject to a minimum probability of detection by a warden.
Thus, steganography has two competing goals: embedding rate and undetectability.

Undetectability is concerned with determining how to alter a cover object to embed data without detectable distortion to the object.
The distortion is measured by a cost function which can be modeled as the difference between the cover and stego objects \cite{aly.2011},
or the sensitivity of the warden to changes to the cover object \cite{f5}.
The embedding rate is defined as the ratio of the number of hidden data bits to the number of cover data bits for a given undetectability \cite[Ch. 4]{steg.book}.
Maximizing the embedding rate for a given undetectability is fundamental to any steganographic algorithm.
Multimedia objects are excellent choices for steganography due to their rich structural features and their ubiquity in business and daily life \cite{jenifer2018survey,tew.2014}.

A relation between undetectability and embedding rate can be used to evaluate the performance of a multimedia steganographic approach,
but this requires a suitable cost function.
Multimedia has been shown to follow the Gibbs distribution \cite{gibbs1984, gibbs1993, bovik3}, which is a complex statistical model that makes analysis difficult.
For example, in \cite[Ch. 13]{steg.book} the Kullback–Leibler (KL) divergence (also called relative entropy), was used as a measure of cover distortion,
however this requires a precise model of the cover object.

Previous steganography cost functions are either not accurate for all multimedia types or are accurate for only a certain class of embedding or detection techniques \cite{quantized.gauss}.
Further, the complexity of the Gibbs distribution has resulted in the use of relaxed distributions such as the Gaussian distribution \cite{limits1,limits2,quantized.gauss,Jessica2015,var.est,Jessica2013,markov2,markov8}, or a presumed preposition on the maximum embedding rate for which the
correctness must be demonstrated \cite{limits1, limits2}.
Thus, the corresponding results may not reflect the true performance because of these approximations.

To partially address the above problems,
a constrained optimization problem was proposed in \cite{myconf} that models the undetectability using KL divergence and the embedding rate with
the mutual information between the cover object and stego object distributions.
This problem was solved to obtain an analytic relation between the embedding rate and undetectability.
However, a relaxed statistical model (Multivariate Quantized Gaussian Distribution (MQGD)), was used for the cover and stego objects.
Thus, an achievability proof that verifies the reliable extraction of the embedded message at the receiver when the exact cover realization is not present
and the probability of decoding error is nonzero was not obtained.

In this paper, a constrained optimization problem that provides an upper bound on the maximum embedding rate that a multimedia steganographic method can achieve reliably
(i.e., with low probability of error), with an achievability proof for a given undetectability.
The cover object is modeled using an equivalence between the Gibbs distribution and the Correlated Multivariate Quantized Gaussian Distribution (CMQGD).
The corresponding optimization problem is solved to obtain an analytic form for the maximum embedding rate in terms of the WrightOmega function.
It is also shown that the maximum embedding rate agrees with the commonly used Square Root Law (SRL) for steganography~\cite[Ch. 13]{steg.book}.
Finally, this theoretical maximum embedding rate is compared with other rates in the literature \cite{quantized.gauss}, \cite{markov2}, \cite{markov8}.
The results obtained demonstrate that the upper bound presented here is small compared to the embedding rates for the steganographic methods.
One reason is that this bound is for an optimal detector, which is not the case in the literature.

The rest of this paper is organized as follows.
Section~\ref{sec:prev} presents the related work.
The main definitions and assumptions for the proposed constrained optimization problem are given in Section \ref{sec:model}.
Section \ref{sec:MultimediaStoch} provides an equivalence relation between the Gibbs distribution and CMQGD.
The solution of the proposed constrained optimization problem is given in Section \ref{sec:proof} and the achievability proof is presented in Section \ref{sec:ach}.
Section \ref{sec:srl} considers the relationship between the SRL and results given here.
An interpretation is given in Section \ref{sec:exp} and discussed in Section~\ref{sec:discuss}.
Finally, some concluding remarks are given in Section \ref{sec:conc}.

\section{Related Work}\label{sec:prev}
Previous steganography approaches employed the KL divergence as a cost function and derived a theoretical bound (under different assumptions), on the
maximum number of bits that can be embedded in a class of multimedia cover objects for a specified probability of detection by a warden.
In \cite{limits1}, continuous Gaussian distributed cover objects were considered (AWGN channels), and the achievability and converse were derived.
The technique in \cite{limits1} was generalized in \cite{limits2} for Multivariate Continuous Gaussian Distributed (MCGD) cover objects (MIMO Gaussian channels).
Although these approaches employ an accurate  model for the cover objects  with the KL divergence cost function,
a preposition on the maximum embedding rate is assumed.

A systematic technique for constructing the distortion function based on the relationship between the steganographic Fisher information and KL divergence
was given in \cite{Jessica2013}.
The embedded data was modeled as an additive image distortion,
and individual pixel costs were determined adaptively to minimize the KL divergence for embedding using least-significant bit matching.
A zero mean independent MQGD was used as the cover object model considering the noise from imaging sensors.
It was shown that the security is better than with the Highly Undetectable Steganograpy (HUGO) algorithm \cite{hugo}.
These results were improved in \cite{Jessica2015} using a Multivariate Generalized Gaussian (MVGG) model combined with a new variance estimator.
This improved the steganographic performance by allowing larger amplitude embeddings in complex (highly textured) image regions due to the
larger tail of the MVGG model.
Content-adaptive pentary steganography utilizing a MVGG model provides comparable performance to
pentary coded HiLL \cite{HILL} and S-UNIWARD \cite{SUNIWARD} against the maxSRMd2 \cite{maxSRMd2} and SRM \cite{SRM}
feature-based steganalysis methods.
Although these techniques provide an accurate analytic relation between the embedding rate and cost function,
the results are based on a relaxed cover model.

In \cite{markov2}, a Gaussian Markov Random Field (GMRF) model was employed which is a non-additive model with a four-element cross neighborhood.
This model has low-dimensional clique structures that can capture the interdependencies among spatially contiguous pixels.
The cost function minimizes the KL divergence between the cover object and stego object.
The cover image is split into two disjoint sub-images which are conditionally independent.
Then, an alternating iterative optimization technique is employed to achieve an efficient embedding which minimizes the KL divergence.
Results were presented which show that the proposed approach outperforms the MiPOD \cite{var.est} and HiLL techniques in terms of
the secure payload against SRM and maxSRMd2.
This approach was extended in \cite{markov8} to higher dimension clique structures (eight element GMRF).
Although analytic relationships between the embedding rate and cost function were given in \cite{markov2,markov8},
they are based on a relaxed model for the cover object.

The approach in \cite{quantized.gauss} considers gray-scale images as cover objects and assumes a MQGD for the cover object, hidden message, and stego object.
Theoretical bounds for the Bayes, Minimax, and Neyman-Pearson likelihood ratio tests were obtained.
This approach maximizes the steganalysis error related to a specified payload.
Although an analytic relation between the embedding rate and cost function was given, the scope of the cost function is limited to
only three types of detectors and a relaxed model (MQGD), is used.

The Square Root Law (SRL)~\cite[Ch. 13]{steg.book} is regarded as the first analytic model of the relationship between the embedding rate
and undetectability in a global framework for multimedia.
The SRL states that the embedding rate for any steganographic technique is proportional to the square root of the number of cover elements.
For example, assume cover object $A$ contains $x$ elements and cover object $B$ contains $4x$ elements.
If we can embed $y$ bits of a secret message within $A$ with a probability of detection $\mathcal{P}$ using a given technique then for the same probability of detection
we can only embed $2y$ bits within $B$.
Further, the embedding rate per cover element in $A$ is $\frac{y}{x}$, but for $B$ the corresponding embedding rate is $\frac{y}{2x}$.
The SRL is valid for any steganographic approach regardless of the embedding and steganalysis methods.
However, it is only a proportional result for a given cost function and so the scaling constant is unknown.

\section{System Model and Problem Statement}\label{sec:model}

\begin{table}
  \centering
  \caption{List of Symbols}
  \label{tab.ListOfSymbols}
  \begin{tabular}{ c | l }
  \hline
$\Sender$ & Sender  \\ \hline
$\Receiver$ & Receiver  \\ \hline
$\Eavesdropper$ & Eavesdropper  \\ \hline
$\AlphabetCovSteg$ & Alphabet of cover and stego elements  \\ \hline
$\AlphabetMsg$ & Alphabet of the original message  \\ \hline
$\AlphabetCodedMsg$ & Alphabet of the coded message  \\ \hline
$\CoverObject$ & Cover object  \\ \hline
$\StegObject$ & Stego object  \\ \hline
$\Msg$ & Original message  \\ \hline
$\CodedMsg$ & Coded message  \\ \hline
$P_D$ & Probability of steganalysis detection at $\mathbf{E}$  \\ \hline
$\TotalProbErrorEavesdropper$ & Probability of steganalysis error at $\mathbf{E}$  \\ \hline
$\AvgProbErrorEavesdropper$ & Average probability of steganalysis error at $\mathbf{E}$ \\ \hline
$\ProbErrorReceiver$ & Probability of decoding error at $\mathbf{B}$ \\ \hline
$\CoverDist$ & Joint probability distribution of the cover object \\ \hline
$\CoverDistc$ & Continuous version of $\CoverDist$ \\ \hline
$\StegDist$ & Joint probability distribution of the stego object \\ \hline
$\StegDistc$ & Continuous version of $\StegDist$  \\ \hline
$\MsgDist$ & Joint probability distribution of the coded message \\ \hline
$\MsgDistc$ & Continuous version of $\MsgDist$ \\ \hline
$\CoverElement$ & $i$th cover element  \\ \hline
$\StegElement$ & $i$th stego element  \\ \hline
$\CoverElementDist$ & Probability distribution of the $i$th cover object  \\ \hline
$\CoverElementDistc$ & Continuous version of $\CoverElementDist$  \\ \hline
$\StegElementDist$ & Probability distribution of the $i$th stego object  \\ \hline
$\StegElementDistc$ & Continuous version of $\StegElementDist$ \\ \hline
$\Sigma_c$, $\Sigma_s$, $\Sigma_m$ & Covariance matrices for $\CoverDist$, $\StegDist$ and $\MsgDist$, respectively, \\
& and also valid for $\CoverDistc$, $\StegDistc$ and $\MsgDistc$, respectively \\ \hline
$\mathrm{tr}(X)$ & Trace of matrix $X$\\ \hline
  \end{tabular}
\end{table}

Elements of a cover object are assumed to be the communication channel between two entities: sender $\Sender$ and receiver $\Receiver$, and
an eavesdropper $\Eavesdropper$ monitors this channel.
Fig. \ref{fig:steg} gives the general communication model for steganography.
In this paper, we consider an innocent cover object $\CoverObject=[c_1, c_2, \dots , c_n]$
and a stego object $\StegObject=[s_1, s_2, \dots s_n]$ with probability distributions $\CoverDist$ and $\StegDist$, respectively,
where $n$ is the number of elements in the cover and stego objects and
$c_i,s_j \in \AlphabetCovSteg$ where $\AlphabetCovSteg$ is the set of allowed cover values.
For example, if the cover object is a gray-scale image, then $\AlphabetCovSteg = \{0,1, \dots ,255\}$ is the set of possible pixel values.
Note that $\AlphabetCovSteg$ is not limited to pixels, but can also be, for example, discrete cosine transform (DCT) coefficients \cite{dccoef}
or video motion vectors \cite{aly.2011,mvs}.
Similarly, we consider a message $\Msg \in \AlphabetMsg^{k}$, where $k$ is the number of message elements and $\AlphabetMsg$ is the set of possible
message values.
The stego object $\StegObject$ is obtained by embedding a coded message $\CodedMsg \in \AlphabetCodedMsg^{n}$ via an addition process, i.e.
$\StegObject = \CoverObject + \CodedMsg$.
The coded message $\CodedMsg$ is obtained from $\Msg$ using a codebook that maps $\Msg \mapsto \CodedMsg$ and has
the same probability distribution as $\MsgDist$.

\begin{figure}
  \centering
  \resizebox{0.5\textwidth}{!}{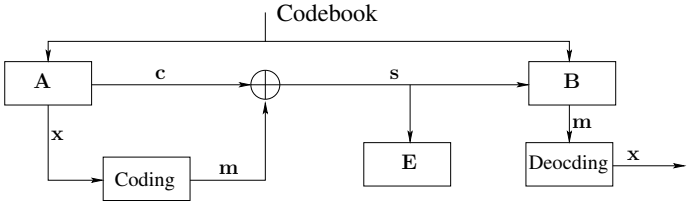}
  \caption{The steganography communication model.}
  \label{fig:steg}
\end{figure}

\subsection{Problem Statement}\label{subsec:problem}

Consider the scenario in which $\Sender$ has access to a multimedia cover object $\CoverObject$ drawn from a distribution $\CoverDist$ which is known
by both $\Receiver$ and $\Eavesdropper$.
To conceal the secret information within $\CoverObject$, $\Sender$ transmits a perturbed version of $\CoverObject$ called the stego object $\StegObject$
and the perturbations are known to only $\Receiver$.
The goal is for $\Sender$ to minimize the probability of detection, $P_D$, by $\Eavesdropper$.

The fundamental problem is then to quantify for a given $\CoverDist$ how much information can be exchanged reliably between $\Sender$ and $\Receiver$
while maintaining a given $P_D$ by $\Eavesdropper$.
Since $\Sender$ does not know the detection strategy at $\Eavesdropper$, the transmission strategy should be designed
considering an optimal detector at $\Eavesdropper$.
In addition, no prior knowledge of the exact realization of the cover object is assumed at $\Receiver$ or $\Eavesdropper$.
In other words, the cover object realization is selected at random by $\Sender$.

We assume that $\Eavesdropper$ employs a binary hypothesis test for $\mathcal{H}_0$ and $\mathcal{H}_1$
where $\mathcal{H}_0$ means $\Sender$ has not embedded a secret message and $\mathcal{H}_1$ means $\Sender$ has embedded a secret message.
Thus, there are two types of errors.
\begin{itemize}
  \item Type $\mathbf{I}$ error: Deciding $\mathcal{H}_1$ when $\mathcal{H}_0$ is true, called false positive (false alarm).
  We denote the probability of this type of error by $\alpha$.
  \item Type $\mathbf{II}$ error: Deciding $\mathcal{H}_0$ when $\mathcal{H}_1$ is true, called false negative (miss detection).
  We denote the probability of this type of error by $\beta$.
\end{itemize}

Assuming equal priors for the optimal hypothesis test at $\Eavesdropper$,
the total probability of error is
\begin{equation}\label{eq.pe_pd}
  \TotalProbErrorEavesdropper = 1- P_D = \alpha + \beta = 2\AvgProbErrorEavesdropper,
\end{equation}
where $\AvgProbErrorEavesdropper$ is the average probability of error for $\Eavesdropper$ for equal priors.
According to \cite{lehmann2006testing}
\begin{equation}\label{eq.pe}
  \TotalProbErrorEavesdropper = 1 - \mathcal{V}(\StegDist,\CoverDist),
\end{equation}
where $\mathcal{V}(\StegDist,\CoverDist)$ is the total variation distance between $\StegDist$ and $\CoverDist$ defined as \cite{lehmann2006testing}
\begin{equation}\label{eq.v}
  \mathcal{V}(\StegDist,\CoverDist) = \dfrac{1}{2}\norm{\StegDist-\CoverDist}_1,
\end{equation}
and $\norm{.}_1$ is the $L_1$ norm.
As an analytic form of the $L_1$ norm is typically complex,
the KL divergence can be employed \cite[Ch. 11]{info.th.book} to obtain
\begin{equation}\label{eq.klv}
  \mathcal{V}(\StegDist,\CoverDist) \leq \sqrt{\dfrac{1}{2} \mathcal{D}(\StegDist \parallel \CoverDist)},
\end{equation}
where
\begin{align}\label{eq.kl}
 \mathcal{D}(\StegDist \parallel \CoverDist) & = \sum_{j=1}^{\dot{n}} P_{s_j} \ln \frac{P_{s_j}}{P_{c_j}},
\end{align}
and $P_{s_j}$ and $P_{c_j}$ are the probability mass function bins for the cover and stego objects, respectively,
and $\dot{n}$ is the number of bins.
To guarantee a low detection probability at $\Eavesdropper$, $\Sender$ must bound $\mathcal{V}(\StegDist,\CoverDist)$ by some $\epsilon$
based on the desired probability of detection.
From (\ref{eq.klv}), $\Sender$ should design a steganographic technique such that
\begin{equation}
\label{eq.c1}
2 \times   \mathcal{V}(\StegDist,\CoverDist)^2 \leq \mathcal{D}(\StegDist \parallel \CoverDist) \leq 2\epsilon^2.
\end{equation}
On the other hand, $\Receiver$ must decode the message $\Msg$ from $\CodedMsg$ sent by $\Sender$.
Thus, the codebook must be selected to maximize the mutual information between $\StegDist$ and $\MsgDist$
\begin{equation}\label{eq.mi}
  I(\MsgDist ; \StegDist) = H(\MsgDist) - H(\MsgDist\mid \StegDist),
\end{equation}
where $H(\MsgDist)$ is the entropy of $\MsgDist$ and $H(\MsgDist\mid \StegDist)$ is the conditional entropy of $\MsgDist$ given $\StegDist$.
Thus, the proposed constrained optimization problem can be written as
\begin{equation}\label{eq.constrain_org}
  \underset{\MsgDist}{\mathrm{argmax}} \quad I(\MsgDist ; \StegDist) \qquad  \mathrm{s.t.} \qquad  \mathcal{D}(\StegDist \parallel \CoverDist) \leq 2\epsilon^2.
\end{equation}
Note that $I(\MsgDist ; \StegDist)$ in (\ref{eq.mi}) can be written as~\cite[Ch. 2]{info.th.book}
\begin{equation}\label{eq.mi2}
  I(\MsgDist ; \StegDist) = H(\StegDist) - H(\StegDist\mid \MsgDist) = H(\StegDist) - H(\CoverDist),
\end{equation}
and because $H(\CoverDist)$ is known, maximizing $I(\MsgDist ; \StegDist)$ is equivalent to maximizing $H(\StegDist)$.
Thus, the proposed optimization problem in \eqref{eq.constrain_org} can be rewritten as
\begin{equation}\label{eq.constrain}
\underset{\MsgDist}{\mathrm{argmax}} \quad H(\StegDist) \qquad  \mathrm{s.t.} \qquad  \mathcal{D}(\StegDist \parallel \CoverDist) \leq  2\epsilon^2.
\end{equation}

\section{Stochastic Multimedia Model}\label{sec:MultimediaStoch}

In this section, the multimedia distribution model, $\CoverDist$, is given.
We first define Markov Random Field (MRF) and Gibbs distribution.

MRF is a multidimensional random process which generalizes the single dimensional Markov random process \cite{bovik3}.
Let $\ddot{X}$ be a coordinate system in $R^N$ and $\rho(i)$ a function representing the neighbourhood of element $i\in \ddot{X}$
such that $i \notin \rho(i)$ and $i \in \rho(j) \mbox{ iff } j \in \rho(i)$.
For example, the neighbourhood can be defined as the immediate left, right, top and bottom neighbours of $i$.
Let $\ddot{Y}$ be the neighbourhood system representing the set of neighbourhoods of all elements in $\ddot{X}$.

A random field $\xi$ over $\ddot{X}$ is a multidimensional random process where each element $i\in \ddot{X}$ is assigned a random variable $\xi_i$ and
$\mathfrak{f}_i$ is the associated realization for all $i \in \ddot{X}$.
$\xi$ is called a MRF if it satisfies the following conditions
\begin{itemize}
\item Positivity: $P(\xi=\mathfrak{f}) > 0$, $\forall \mathfrak{f} \in \ddot{S}$, and
\item Markovianity: $P(\xi_i=\mathfrak{f}_i | \xi_j=\mathfrak{f}_j , \forall j \neq i) = P(\xi_i=\mathfrak{f}_i | \xi_j=\mathfrak{f}_j , \forall j \in \rho(i) ), \forall i \in \ddot{X}, \forall \mathfrak{f} \in \ddot{S}$,
\end{itemize}
where $P$ is the probability measure and $\ddot{S}$ is the state space for $\xi$.

The Hammersley-Clifford Theorem \cite{bovik3} states that $\xi$ is a MRF over $\ddot{X}$ with respect to $\ddot{Y}$ if and only if its probability distribution
follows the Gibbs distribution with respect to $\ddot{X}$ and $\ddot{Y}$.
To explain the Gibbs distribution, a clique must be defined.
A clique $\omega$ is a correlated group of neighbouring elements where $\omega\subset\ddot{X}$ with respect to $\ddot{Y}$, i.e.
$\omega$ consists of a single element or multiple elements which are neighbours.
Further, $\omega \in \ddot{\Omega}$ where $\ddot{\Omega}$ denotes the set of all cliques.

According to \cite{gibbs1984, gibbs1993, bovik3}, multimedia objects can be modeled using the Gibbs distribution.
This implies the above described properties (positivity and Markovianity), so the elements of a multimedia object can be organized in cliques.
The Gibbs distribution employs an energy function $U_\mathfrak{f}$ and a temperature constant $\hat{T}$ to obtain the probability
\begin{equation}
P(\xi=\mathfrak{f}) = \frac{1}{Z} e^{- U_\mathfrak{f} / \hat{T}},
\label{eq.gibbs}
\end{equation}
where
\begin{equation}
Z = \sum_{k=1}^{M}{e^{- U_k / \hat{T}}},
\end{equation}
is the normalization denominator called the partition function and $M$ is the number of system states (number of elements in $\ddot{S}$).
The energy function is modeled as
\begin{equation}\label{eq.gibbs.energy.func}
U_\mathfrak{f} = \sum_{\omega \in \ddot{\Omega}} V_\mathfrak{f}(\omega),
\end{equation}
where $V_\mathfrak{f}(\omega)$ is the potential function for $\mathfrak{f}$ calculated only within the clique $\omega \in \ddot{\Omega}$.

The Gibbs distribution requires the partitioning function which is intractable to obtain for most models.
Therefore, a CMQGD is used here to model the Gibbs distribution.
This is implemented by modeling \eqref{eq.gibbs} with a CMQGD by modeling $V_f(\omega)$ in \eqref{eq.gibbs.energy.func} as
\begin{equation}\label{eq.PotFun}
V_\mathfrak{f}(\omega) = \frac{\hat{T}}{2}(\mathfrak{f}_\omega - \mu_\omega)\Sigma_\omega^{-1}(\mathfrak{f}_\omega - \mu_\omega)^T,
\end{equation}
where $\mathfrak{f}_\omega$ is the realization of the random field $\xi$ within clique $\omega$,
i.e. $\mathfrak{f}_i \in \mathfrak{f}_\omega \quad \forall i \in \omega$, $\mathfrak{f}_\omega \subset \mathfrak{f}$, and
$\mu_\omega$ and $\Sigma_\omega$ are the mean vector of $\xi_\omega \subset \xi$ and the corresponding covariance matrix within the clique $\omega$, respectively.
This equivalence is valid for a wide class of Gibbs distributions in which the potential function can be expressed in a quadratic form of the exponent of an equivalent
Gaussian distribution.
This is based on the following.
\begin{itemize}
\item $V_\mathfrak{f}(\omega)$ can be mapped to any function using only the elements within the clique $\omega$ \cite{bovik3}.
\item The product of $(\mathfrak{f}_\omega - \mu_\omega)$ and $(\mathfrak{f}_\omega - \mu_\omega)^T$ (the distance from the mean for clique $\omega$),
controls the probability of a given realization (state) $\mathfrak{f}_\omega$,
and hence whether $P(\xi=\mathfrak{f})) \in \omega$ for a given $\Sigma_\omega$.
In other words, the local coherence condition of multimedia elements, $P(\xi=\mathfrak{f}) \propto -U_\mathfrak{f} \propto -V_\mathfrak{f}(\omega)$,
is satisfied.
This results in low energy states having a higher probability than high energy states, which is one of the properties of the Gibbs distribution \cite{gibbs1984}.
\item $Z$, which is the normalization denominator, is equal to
\[
\sqrt{2\pi}|\Sigma_\omega| = \sum_{\omega\in \ddot{\Omega}} e^{-\sum{\frac{\hat{T}}{2}(\mathfrak{f}_\omega
- \mu_\omega)\Sigma_\omega^{-1}(\mathfrak{f}_\omega - \mu_\omega)^T}}.
\]
\end{itemize}
It should be noted that other assumptions for the potential function $V_\omega(\mathfrak{f})$ can be used to model the Gibbs distribution.

\section{Main Results}\label{sec:proof}
In this section, the solution of the optimization problem~\eqref{eq.constrain} in Section~\ref{sec:model} is presented using the assumptions given below.
\begin{itemize}
\item Each clique within the cover is considered as a single cover element $\CoverElement$ with unknown distribution $\CoverElementDist$.
\item The equivalence relation between Gibbs and Gaussian distributions in Section \ref{sec:MultimediaStoch} is used to
model the cover object, i.e. the joint distribution of $\CoverDist$, as a CMQGD with mean $\vec{\mu_c}$ and covariance matrix $\Sigma_c$.
No assumptions are made on the marginal distribution $\CoverElementDist$ of $\CoverElement$.
\item The cover statistical parameters ($\vec{\mu_c}$,$\Sigma_c$) are known to $\Sender$, $\Receiver$, and $\Eavesdropper$.
\item Although no assumptions are made on $\StegDist$ and $\MsgDist$,
minimizing the KL divergence for a Gaussian distributed cover object requires that the corresponding stego object also be Gaussian distributed \cite{gaussisgauss}.
Thus, the embedding operation will produce a CMQGD $\StegDist$ with parameters ($\vec{\mu_s}$, $\Sigma_s$).
Although the zero mean Gaussian distribution case is considered in \cite{gaussisgauss}, the proof can be extended to any Gaussian distribution
as shown in Appendix A.
\item The codebook realization from $\MsgDist$, i.e. the codewords used between $\Sender$ and $\Receiver$,
is shared only between $\Sender$ and $\Receiver$ and is not known to $\Eavesdropper$.
\end{itemize}

\subsection{Solution of \eqref{eq.constrain}}\label{subsec.solution}
To solve~\eqref{eq.constrain}, we must determine $H(\CoverDist)$, $H(\StegDist)$ and the KL divergence $\mathcal{D}(\StegDist \parallel \CoverDist)$.
First, the KL divergence between two quantized distributions is evaluated using \eqref{eq.kl.epsilon} and the following lemma.
\begin{lemma}\cite[Ch. 11]{info.th.book}\label{TH1}
The KL divergence between two uniformly quantized distributions $F(\dot{X})$ and $G(\dot{X})$ is bounded as
	\begin{equation}\label{eq.th1}
	\mathcal{D}(F(\dot{X}) \parallel G(\dot{X})) \leq \mathcal{D}(f(\dot{x}) \parallel g(\dot{x}))
	\end{equation}
where $f$ and $g$ are continuous versions of $F$ and $G$, respectively, and $\dot{X}$ is the uniformly quantized version of the continuous random variable $\dot{x}$.
\end{lemma}

From Lemma~\ref{TH1} and the definition of KL divergence for a MCGD \cite{limits2}, the KL divergence for a CMQGD can be bounded as
\begin{align}\label{eq.kl:qgd}
\mathcal{D}(\StegDist \parallel \CoverDist) & \leq \frac{1}{2} \bigg(\mathrm{tr}(\Sigma_c^{-1} \Sigma_s) + (\vec{\mu_c} - \vec{\mu_s})^T
\Sigma_c^{-1} (\vec{\mu_c} - \vec{\mu_s}) \nonumber \\
& + \ln \frac{|\Sigma_c|}{|\Sigma_s|} - n\bigg),
\end{align}
where $n$ is the number of cover elements.
From \eqref{eq.kl:qgd}, let $\ddot{A}=\Sigma^{-1}_c \Sigma_s$, then we have
\begin{equation*}
|\ddot{A}| = |\Sigma^{-1}_c| |\Sigma_s| = \frac{|\Sigma_s|}{|\Sigma_c|}.
\end{equation*}
As $\Sigma^{-1}_c$ and $\Sigma_s$ are symmetric and positive semi-definite
\begin{align*}
|\ddot{A}| &= \prod\lambda_i, \quad \mathrm{tr}(\ddot{A}) = \sum\lambda_i,
\end{align*}
and $\mathrm{tr}(\ddot{A})$ is monotonically increasing with $|\Sigma_s|$
where $\sum\lambda_i$ and $\prod\lambda_i$ are the sum and product of the eigenvalues of $\ddot{A}$, respectively.
Then from the arithmetic-geometric inequality we have
\begin{equation*}
\frac{1}{\ddot{n}}\mathrm{tr}(\ddot{A}) \geq |\ddot{A}|^{\frac{1}{\ddot{n}}},
\end{equation*}
where $\ddot{n}$ is the number of eigenvalues.
As $x>\ln(x)$ for any positive real number $x$
\begin{align}\label{proof.D.Sigma}
\frac{1}{\ddot{n}}\mathrm{tr}(\ddot{A}) &> \frac{1}{\ddot{n}}\ln(|\ddot{A}|),\nonumber \\
\mathrm{tr}(\ddot{A}) &> \ln(|\ddot{A}|),\nonumber \\
\mathrm{tr}(\Sigma_c^{-1} \Sigma_s) &> \ln \left(\frac{|\Sigma_s|}{|\Sigma_c|}\right),
\end{align}
so \eqref{eq.kl:qgd} can be rewritten as
\begin{align}\label{eq.kl:qgd.new}
\mathcal{D}(\StegDist \parallel \CoverDist) &\leq \frac{1}{2} \bigg(\mathrm{tr}(\Sigma_c^{-1} \Sigma_s) + (\vec{\mu_c} - \vec{\mu_s})^T \Sigma_c^{-1} (\vec{\mu_c} - \vec{\mu_s}) \nonumber \\
 &- \ln \frac{|\Sigma_s|}{|\Sigma_c|} - n\bigg).
\end{align}
Then from \eqref{eq.th1}, \eqref{proof.D.Sigma}, and \eqref{eq.kl:qgd.new},
$\mathcal{D}(\StegDist \parallel \CoverDist)$ and $\mathcal{D}(\StegDistc \parallel \CoverDistc)$ are monotonically increasing with $|\Sigma_s|$,
and from \eqref{eq.ent.mqgd} and \eqref{eq.ent.ps}, $H(\StegDist)$ and $h(p_s)$ are also monotonically increasing with $|\Sigma_s|$,
where $h(p_s)$ is the differential entropy of the continuous CMQGD $p_s$ and $\Sigma_s$ is the only variable in \eqref{eq.ent.ps}
as $2\pi e$ and $b$ can be considered as constants.

To evaluate the optimization problem in \eqref{eq.constrain},
we use the upper bound on $\mathcal{D}(\StegDist \parallel \CoverDist)$,
i.e. $\mathcal{D}(\StegDist \parallel \CoverDist) = \mathcal{D}(\StegDistc \parallel \CoverDistc))$, so that
\begin{align}\label{eq.kl.epsilon}
\mathcal{D}(\StegDistc \parallel \CoverDistc)& = \frac{1}{2} \bigg(\mathrm{tr}(\Sigma_c^{-1} \Sigma_s) + (\vec{\mu_c} - \vec{\mu_s})^T \Sigma_c^{-1} (\vec{\mu_c} - \vec{\mu_s})\\
&+ \ln \frac{|\Sigma_c|}{|\Sigma_s|} - n\bigg).
\end{align}
Then, substituting \eqref{eq.kl.epsilon} in \eqref{eq.constrain}, the Lagrangian of \eqref{eq.constrain} is
\begin{equation}\label{eq.lag}
\mathcal{L}(\StegDistc, \lambda) = \mathcal{L}(\vec{\mu_s}, \Sigma_s, \lambda) = h(\StegDistc) - \lambda(\mathcal{D}(\StegDistc \parallel \CoverDistc) - 2\epsilon^2),
\end{equation}
where $\lambda$ is the Lagrange multiplier.

The solution of the optimization problem in~\eqref{eq.lag} is given in the following theorem.
\begin{theorem}\label{TH2}
With the given assumptions, the solution of \eqref{eq.lag} is obtained as follows.
\begin{enumerate}
\item The distribution of $\MsgDist$ must be CMQGD for both the cover and stego objects, with mean and covariance matrix
	\begin{enumerate}
	\item $\vec{\mu_m} = 0$, and
	\item $\Sigma_m = \Big( -W(-\frac{4\epsilon^2}{n} -1 -i\pi) -1 \Big)\Sigma_c$,
	\end{enumerate}
for steganalysis with $P_D \leq \sqrt{\frac{\mathcal{D}(\StegDist \parallel \CoverDist)}{2}} \leq \epsilon)$
where $W(.)$ is the WrightOmega function \cite{WrightOmega}.
\item The maximum achievable embedding rate is
\[
I(\StegDist;\MsgDist) = \frac{n}{2}\ln(-W(-\frac{4\epsilon^2}{n} -1 -i\pi)),
\]
for sufficiently large $n$.
\end{enumerate}
 \end{theorem}
\begin{proof}
To find the optimal parameters for \eqref{eq.lag},
set $\frac{\partial}{\partial \vec{\mu_s}}\mathcal{L} = 0$ and $\frac{\partial}{\partial \Sigma_s}\mathcal{L} = 0$.
With the definitions of $h(\StegDistc)$ and $\mathcal{D}(\StegDist \parallel \CoverDist)$ in \eqref{eq.ent.mqgd} in Appendix B and \eqref{eq.kl.epsilon},
respectively, we have
\begin{align}
\label{eq.start1} \frac{\partial}{\partial \vec{\mu_s}} h(\StegDistc) & =  0 , \\
\label{eq.start2} \frac{\partial}{\partial \Sigma_s} h(\StegDistc) & =  \frac{1}{2}(\Sigma_s^{-1})^T ,\\
\label{eq.start3.5} \frac{\partial}{\partial \vec{\mu_s}} \mathcal{D}(\StegDistc \parallel \CoverDistc) & =  \frac{-1}{2}\Big(\Sigma_c^{-1} (\vec{\mu_c} - \vec{\mu_s}) + (\Sigma_c^{-1})^T (\vec{\mu_c} - \vec{\mu_s}) \Big), \\
\label{eq.start4} \frac{\partial}{\partial \Sigma_s} \mathcal{D}(\StegDistc \parallel \CoverDistc) & = \frac{1}{2} \big( \Sigma_c^{-1} - \Sigma_s^{-1} \big).
\end{align}
As $\Sigma_c$ and $\Sigma_s$ are symmetric matrices, \eqref{eq.start3.5} can be rewritten as
\begin{align}
\label{eq.start3} \frac{\partial}{\partial \vec{\mu_s}} \mathcal{D}(\StegDistc \parallel \CoverDistc) & =  -\Sigma_c^{-1} (\vec{\mu_c} - \vec{\mu_s}).
\end{align}
From \eqref{eq.start1} and \eqref{eq.start3}, as $\frac{\partial}{\partial \vec{\mu_s}}\mathcal{L} = 0$, we have
\begin{equation}\label{eq.result1}
\vec{\mu_s} = \vec{\mu_c}.
\end{equation}
Further, from \eqref{eq.start2} and \eqref{eq.start4}, as $\frac{\partial}{\partial \Sigma_s}\mathcal{L} = 0$, we have
\begin{align}\label{eq.result2.a}
(\lambda +1) \Sigma_s^{-1} &= \lambda \Sigma_c^{-1}, \nonumber \\
\Sigma_s &= \frac{\lambda +1}{\lambda} \Sigma_c, \nonumber \\
\Sigma_s &= a \Sigma_c.
\end{align}
Let $a = \frac{\lambda +1}{\lambda}$ be the embedding factor.
From~\eqref{eq.result2.a}, it is clear that $\Sigma_s$ has the same eigenvectors as $\Sigma_c$ and the
eigenvalues of $\Sigma_s$ are a scaled version of the eigenvalues of $\Sigma_c$ by the embedding factor $a$.

Let the characteristic functions of the cover and stego objects be
\begin{align*}
\psi_{c}(\vec{t}) &= e^{i\vec{\mu_c}^T \vec{t} + \frac{1}{2}\vec{t}^T\Sigma_c \vec{t}}, \\
\psi_{s}(\vec{t}) &= e^{i\vec{\mu_s}^T \vec{t} + \frac{1}{2}\vec{t}^T\Sigma_s \vec{t}}.
\end{align*}
As $\StegObject = \CoverObject + \CodedMsg$, this implies $\StegDist = \CoverDist \oplus \MsgDist$, where $\oplus$ denotes convolution.
Thus, the characteristic function for the message codebook distribution can be expressed as
\begin{align*}
\psi_{m}(\vec{t}) &=  \frac{\psi_{s}(\vec{t})}{\psi_{c}(\vec{t})} = e^{i(\vec{\mu_s}^T-\vec{\mu_c}^T)\vec{t}+ \frac{1}{2}\vec{t}^T (\Sigma_s - \Sigma_c) \vec{t}}.
\end{align*}
Then from \eqref{eq.result1}, we have
\begin{equation}\label{eq.ch_m}
	\psi_{m}(\vec{t}) = e^{ \frac{1}{2}\vec{t}^T(\Sigma_s - \Sigma_c) \vec{t}} = e^{ \frac{1}{2}\vec{t}^T(\Sigma_m) \vec{t}}.
\end{equation}
and
\begin{equation} \label{eq.mum2}
\mu_m = 0,
\end{equation}
so the message codebook can be modeled as a CMQGD with zero mean and covariance matrix
\begin{align}
\Sigma_m & = \Sigma_s - \Sigma_c  \nonumber \\
				& = ( a -1 )\Sigma_c.\label{eq.sigmam2}
\end{align}

The embedding factor $a$ is obtained as follows.
From \eqref{eq.result2.a}, we have
\begin{equation}\label{eq.result2.b}
\mathrm{tr}(\Sigma_c^{-1} \Sigma_s) = na,
\end{equation}
and
\begin{equation}\label{eq.result2}
\ln \frac{|\Sigma_c|}{|\Sigma_s|} = -n \ln (a).
\end{equation}
Substituting \eqref{eq.result1}, \eqref{eq.result2.b} and \eqref{eq.result2} in \eqref{eq.kl.epsilon} gives
\begin{align}\label{eq.result3}
a n - n \ln(a) &= 2 \mathcal{D}(\StegDistc \parallel \CoverDistc) +n, \nonumber \\
a -\ln(a) &=\frac{2}{n} \mathcal{D}(\StegDistc \parallel \CoverDistc) +1.
\end{align}
The solution of \eqref{eq.result3} can be obtained from \cite{WrightOmega} as
\begin{equation}\label{eq.result4}
a = -W(-\frac{2}{n} \mathcal{D}(\StegDistc \parallel \CoverDistc) -1 -i\pi).
\end{equation}
Then, substituting the upper bound $2\epsilon^2$ for $\mathcal{D}(\StegDistc \parallel \CoverDistc)$ from \eqref{eq.lag} in \eqref{eq.result4}, we have
\begin{equation}\label{eq.result4m}
a =-W(-\frac{4\epsilon^2}{n} -1 -i\pi),
\end{equation}
Equations \eqref{eq.mum2}, \eqref{eq.sigmam2} and \eqref{eq.result4m} complete the proof of part 1).

For part 2), from \eqref{eq.mi2}, and \eqref{eq.ent.pc} and \eqref{eq.ent.ps} in Appendix B, we have
\begin{align}
I(\StegDistc;\MsgDistc)	&=		  h(\StegDistc) - h(\CoverDistc).		\nonumber
\end{align}
Then from \eqref{eq.ent.mqgd} in Appendix B
\begin{align}
I(\StegDistc;\MsgDistc) &=  \frac{1}{2}\ln(2\pi e|\Sigma_s|)- \frac{1}{2}\ln(2\pi e|\Sigma_c|). \nonumber
\end{align}
and from \eqref{eq.sigmam2} and \eqref{eq.result4m}
\begin{align}\label{eq.I} 
I(\StegDistc;\MsgDistc)	&= \frac{n}{2}\ln(a) + \frac{1}{2}\ln(2\pi e|\Sigma_c|)- \frac{1}{2}\ln(2\pi e|\Sigma_c|) \nonumber \\
			&=		  \frac{n}{2}\ln(a).
\end{align}
Thus
\begin{equation}\label{eq.main}
I(\StegDistc;\MsgDistc) = \frac{n}{2}\ln\Big(-W(-\frac{4\,\epsilon^2}{n} -1 -i\pi)\Big).
\end{equation}
which completes the proof of part 2).
\end{proof}

Note that although $\mathcal{D}(\StegDistc \parallel \CoverDistc) \neq \mathcal{D}(\CoverDistc \parallel \StegDistc)$,
Theorem \ref{TH2} is also valid for $\mathcal{D}(\CoverDistc \parallel \StegDistc)$.
The proof of this is given in the following lemma.

\begin{lemma}\label{TH3}
$\vec{\mu_m}$ and $\Sigma_m$ in Theorem \ref{TH2} are also valid for $\mathcal{D}(\CoverDistc \parallel \StegDistc)$.
\end{lemma}
\begin{proof} Equation \eqref{eq.kl.epsilon} can be modified as
\begin{align}\label{eq.kl:qgd.kl2}
\mathcal{D}(\CoverDistc \parallel \StegDistc) & = \frac{1}{2} \bigg(\mathrm{tr}(\Sigma_s^{-1} \Sigma_c) + (\vec{\mu_s} - \vec{\mu_c})^T \Sigma_s^{-1} (\vec{\mu_s} - \vec{\mu_c}) \nonumber \\
& + \ln \frac{|\Sigma_s|}{|\Sigma_c|} - n\bigg).
\end{align}
As a consequence, \eqref{eq.start4} and \eqref{eq.start3} become
\begin{equation}\label{eq.start3.kl2}
\frac{\partial}{\partial \Sigma_s} \mathcal{D}(\CoverDistc \parallel \StegDistc)  = \frac{1}{2} \big( -(\Sigma_s^{-1})^T \Sigma_c (\Sigma_s^{-1})^T + (\Sigma_s^{-1})^T \big),
\end{equation}
\begin{equation}\label{eq.start4.kl2}
\frac{\partial}{\partial \vec{\mu_s}} \mathcal{D}(\CoverDistc \parallel \StegDistc) =  \Sigma_s^{-1} (\vec{\mu_s} - \vec{\mu_c}).
\end{equation}
Then following the same steps as in the proof of Theorem \ref{TH2} using \eqref{eq.start1}, \eqref{eq.start2}, \eqref{eq.start3.kl2}, and \eqref{eq.start4.kl2},
the results in \eqref{eq.mum2}, \eqref{eq.sigmam2}, \eqref{eq.result4}, and \eqref{eq.result4m} in Theorem \ref{TH2} follow.
\end{proof}

\section{Achievability}\label{sec:ach}
In this section, we prove that the maximum embedding rate $I(\StegDist;\MsgDist)$ in~\eqref{eq.main} is achievable
with a low probability of decoding error $\ProbErrorReceiver$ at $\Receiver$.
An achievability constraint is not included in the optimization problem \eqref{eq.constrain} and \eqref{eq.lag} as $\StegDist$ and $\MsgDist$
are obtained after this problem is solved.

The standard random coding argument \cite[Ch. 7]{info.th.book} is used as the basis for the proof.
Although the minimum distance decoder (maximum likelihood estimator), is the optimal detector~\cite[Ch. 7]{info.th.book},
it is difficult to use this detector to calculate $\ProbErrorReceiver$ as we do not have the marginal distribution for each cover element,
only the joint distribution for the entire cover.
Thus, a jointly typical decoder is considered at $\Receiver$.

The approach employed in the achievability proof is the same as in the proof of the Channel Capacity Theorem \cite[Theorem 9.1.1]{info.th.book}, to obtain
$\ProbErrorReceiver$.
Although this theorem constrains the average transmission power,
in our case we do not need this constraint as $\Sigma_m$ is very small due to the low probability of detection required at $\Eavesdropper$ in \eqref{eq.constrain}.
Two types of errors can be defined if the $i$th codeword is transmitted.
\begin{itemize}
\item The transmitted and received codewords are not jointly typical. We denote this error by $\hat{E}_i$.
\item The received codeword is jointly typical with an incorrect codeword.
This error is denoted by $\sum_{j=1, j\ne i}^{\mathcal{K}} E_j$ where $\mathcal{K}$ is the number of codewords.
\end{itemize}
Then (9.26) in \cite{info.th.book} is modified as
\begin{equation}
\ProbErrorReceiver \leq P(\hat{E}_i) + \sum_{j=1, j\ne i}^{\mathcal{K}} P(E_j),
\end{equation}
where $P(x)$ is the probability of event $x$.
Continuing with the approach in \cite{info.th.book} gives
\begin{equation}\label{eq.pb}
\ProbErrorReceiver \leq 2\delta,
\end{equation}
where $\delta >0$ is an arbitrary small number.
Thus, \eqref{eq.pb} proves that when $n$ is sufficiently large and for an embedding rate $R\leq I(\StegDist;\MsgDist)-2\epsilon$,
$\ProbErrorReceiver$ is sufficiently small.
Therefore, a codebook ($\mathcal{K}$,$n$,$\epsilon$) exists that achieves an embedding rate $R\leq I(\StegDist;\MsgDist)-2\epsilon$ at $\Receiver$ with
low $\ProbErrorReceiver$.

\section{Relationship to the SRL}\label{sec:srl}
In this section, we establish the relationship between \eqref{eq.result4m} and \eqref{eq.main} in terms of the WrightOmega function and the
results for the SRL in \cite{steg.book,limits1,limits2}.
According to \cite[Theorem 1]{WrightOmega2}, a special case $\mathcal{W}_{-1} (-e^{-u -1})$ of the Lambert W function $\mathcal{W}(.)$ can be bounded as
\begin{equation}\label{eq.th1.1}
-1 -\sqrt{2u} -u < \mathcal{W}_{-1}(- e^{-u -1}) < -1 -\sqrt{2u} -\frac{2}{3}u.
\end{equation}
From \cite{WrightOmega} and \eqref{eq.lambertw.2} in Appendix C, setting $u=\frac{4\epsilon^2}{n}$ in \eqref{eq.result4m} gives
\begin{align*}
W(-u -1 -i\pi) &= \mathcal{W}_{-1}(e^{-u -1 -i\pi}), \\
 &= \mathcal{W}_{-1}(-e^{-u -1 }),
\end{align*}
and then using \eqref{eq.th1.1}
\begin{align}\label{eq.th1.2}
W(-u -1 -i\pi)  &> -1 -\sqrt{2u} - u, \nonumber \\
-W(-u -1 -i\pi) &<  1 +\sqrt{2u} + u, \nonumber \\
a             &<  1 + \sqrt{\frac{8\epsilon^2}{n}} + \frac{4\epsilon^2}{n}.
\end{align}
As $n$ is a very large number, $\sqrt{\frac{8\epsilon^2}{n}} \gg \frac{4\epsilon^2}{n}$, so
\begin{align}\label{eq.th1.3}
a                <&   \sqrt{\frac{8\epsilon^2}{n}}, \nonumber \\
\sqrt{a}         <&   \sqrt{\frac{8\epsilon^2}{n}}, \nonumber \\
\ln(\sqrt{a})    <&   \sqrt{\frac{8\epsilon^2}{n}}, \nonumber \\
n\ln(\sqrt{a})   <&  n\sqrt{\frac{8\epsilon^2}{n}}, \nonumber \\
n\ln(\sqrt{a})   <&  \sqrt{8\epsilon^2}\sqrt{n}.
\end{align}
From \eqref{eq.I} and \eqref{eq.th1.3}, we have
\begin{equation}\label{eq.Isrl}
I(\StegDistc;\MsgDistc)  < 2\sqrt{2}\epsilon\sqrt{n},
\end{equation}
which establishes the relationship between \eqref{eq.main} and the SRL.
This shows that the results given here are in agreement with those established for the SRL in \cite{steg.book} in the context of image steganography and in
\cite{limits1,limits2} in the context of covert communications.
Based on the justification in Section \ref{sec:MultimediaStoch} for modeling the Gibbs distribution with a CMQGD and \eqref{eq.Isrl},
\eqref{eq.main} can be considered as an analytic form of the SRL.\hfill $\blacksquare$

\section{Bound Interpretation}\label{sec:exp}
In this section, the experimental results for the steganographic methods in \cite{quantized.gauss}, \cite{markov2}, and \cite{markov8} using the BOSSbase 1.01 dataset~\cite{BOSS} are compared with the results given here.
Specifically, the payload (bits per pixel) of each steganographic method is determined versus $\AvgProbErrorEavesdropper$.
Then, these results are compared to the theoretical upper bounds for the payload from \eqref{eq.main}.
As the size of each image of the BOSSbase 1.01 dataset is $512\times 512$ pixels, we set $n = 512 \times 512 = 2^{18}$ in \eqref{eq.main}
to determine the maximum achievable embedding rate for independent cover elements, i.e. when each clique represents only a single pixel,
which is an extreme case for a cover object to provide the maximum achievable embedding rate.
The relationship between the achievable embedding rate and the clique size is investigated in Section~\ref{sec:discuss}.

\Cref{fig:Markov_2_T1a,fig:Markov_2_T1b,fig:Markov_2_T2a,fig:Markov_2_T2b} present the results from \cite{markov2}
and the upper bound in \eqref{eq.I} with $a$ set to the lower bound in \eqref{eq.result4.upper} to obtain an upper bound for the payload versus $\AvgProbErrorEavesdropper$.
Figures \ref{S-fig:Markov_8_T1a,S-fig:Markov_8_T1b,S-fig:Markov_8_T2a,S-fig:Markov_8_T2b} and
Figures \ref{S-fig:quantized_gauss_T1,S-fig:quantized_gauss_T2_HILL,S-fig:quantized_gauss_T2_MiPOD,S-fig:quantized_gauss_T2_SUNIWARD,S-fig:quantized_gauss_T3_HILL,S-fig:quantized_gauss_T3_MiPOD,S-fig:quantized_gauss_T3_SUNIWARD} for the methods in \cite{markov8} and \cite{quantized.gauss}, respectively,
present the average payload (in bits per pixel) for the steganographic methods S-UNIWARD, MIPOD, GMRF BASE, GMRF, MiPOD and HILL
versus $\AvgProbErrorEavesdropper$ for steganalysis using SRM and maxSRMd2.
These results show that the upper bound is small compared to the results for the steganographic methods in the literature.
The reason is that the steganographic detectors employed are not optimal for the steganographic methods.
In other words, the steganalysis methods used in \cite{quantized.gauss,markov2,markov8}
are suboptimal compared to the KL divergence $\mathcal{D}(\StegDistc \parallel \CoverDistc)$ given previously which can be regarded as
an upper bound for any steganographic detector.
Consequently, using a suboptimal detector may be misleading as it will result in a lower probability of detection error, $\AvgProbErrorEavesdropper$,
for a given embedding rate compared to the upper bound.
Further, using a suboptimal embedding technique may result in lower embedding rates than the upper bound for the same detection probability.

\begin{figure}
\centering
	\includegraphics[width=0.75\linewidth, height=0.48\linewidth]{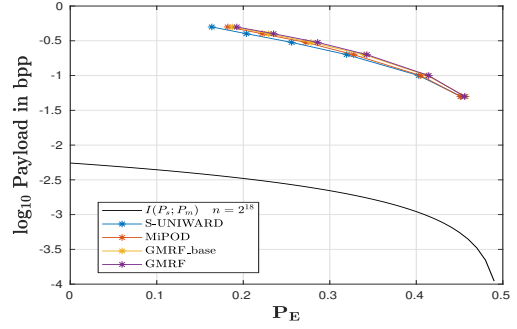}
	    \caption{Payload versus $\AvgProbErrorEavesdropper$ from \cite{markov2} for the steganographic methods S-UNIWARD, MIPOD, GMRF\_BASE and GMRF with steganalysis using SRM.}\label{fig:Markov_2_T1a}
\end{figure}

\begin{figure}\centering
	\includegraphics[width=0.75\linewidth, height=0.48\linewidth]{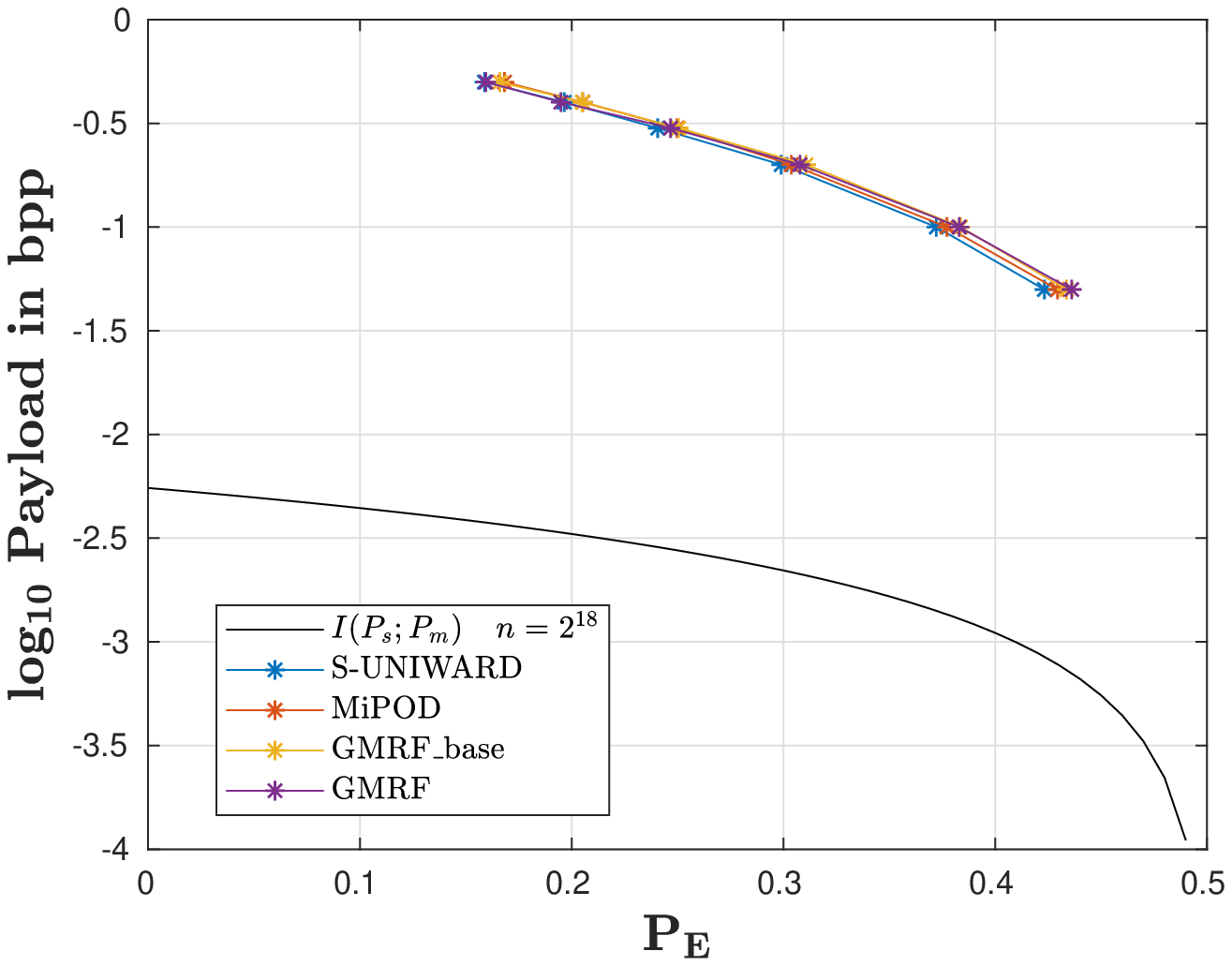}
	    \caption{Payload versus $\AvgProbErrorEavesdropper$ from \cite{markov2} for the steganographic methods S-UNIWARD, MIPOD, GMRF\_BASE and GMRF with steganalysis using maxSRMd2.}\label{fig:Markov_2_T1b}
\end{figure}

\begin{figure}\centering
	\includegraphics[width=0.75\linewidth, height=0.48\linewidth]{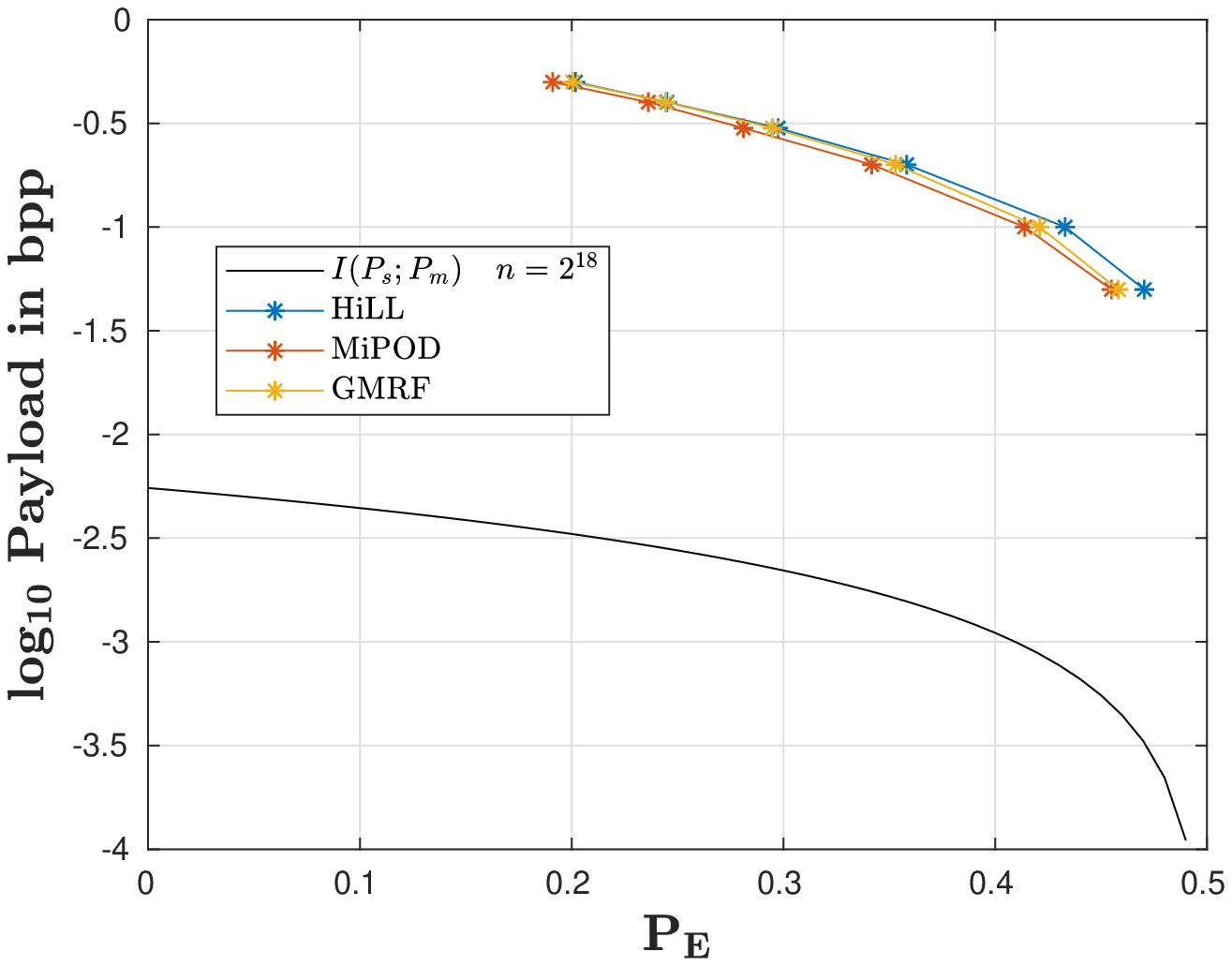}
	    \caption{Payload versus $\AvgProbErrorEavesdropper$ from \cite{markov2} for the steganographic methods MiPOD and GMRF (both enhanced using low-pass filtered cost) and HILL with steganalysis using SRM.}\label{fig:Markov_2_T2a}
\end{figure}

\begin{figure}\centering
	\includegraphics[width=0.75\linewidth, height=0.48\linewidth]{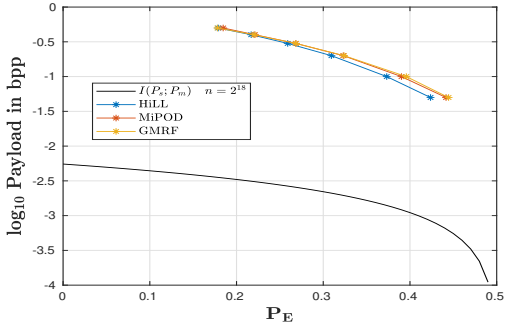}
	    \caption{Payload versus $\AvgProbErrorEavesdropper$ from \cite{markov2} for the steganographic methods MiPOD and GMRF (both enhanced using low-pass filtered cost) and HILL with steganalysis using maxSRMd2.}\label{fig:Markov_2_T2b}
\end{figure}

\section{Discussion}\label{sec:discuss}
From \eqref{eq.Isrl}, $I(\StegDist;\MsgDist)\propto\sqrt{n}$ which implies that
the upper bound for the embedding rate occurs when each clique contains a single element, i.e. when $n$ is large.
This occurs in cases such as a highly textured image or a motion field of a sandstorm video as in Fig. \ref{fig:ex.high}.
Conversely, the lowest embedding rate occurs when all elements of the cover are correlated, i.e. the entire cover can be considered as a single clique so that $n=1$.
This can occur in cases such as an image of a clear sky or a motion field of a stationary video scene with global camera motion as in Fig. \ref{fig:ex.low},
both of which imply a small number of cliques.
In other words, the largest value occurs when all elements of the cover are independent,
whereas the smallest value occurs when the cover contains only a single clique.

It should be noted that the proof of the information-theoretic upper bound for $I(\StegDistc;\MsgDistc)$ in \eqref{eq.main} is not constructive
due to the random coding argument.
In other words, it does not provide a means of designing a steganographic method or codebook to achieve the upper bound.
Only an analytic existence proof was provided here so further work is needed to construct steganographic methods close to the bound.
However, the following can be inferred from Theorem~\ref{TH2} for steganography and steganalysis methods.
\begin{itemize}
\item The achievable embedding rate is monotonically increasing with the variance of the cover elements (or the entropy of the cover),
as shown in Figures \ref{fig:ex.high} and \ref{fig:ex.low}.
\item The achievable embedding rate with respect to a given level of risk, i.e. best steganalysis performance,
is monotonically increasing with the number of cover elements.
\item The message elements should be mapped to a codebook with a distribution similar to that of the cover with zero mean.
\end{itemize}
\begin{figure}
\centering
	\begin{subfigure}{0.8\columnwidth}\centering
	\includegraphics[width=0.8\linewidth, height=0.4\linewidth]{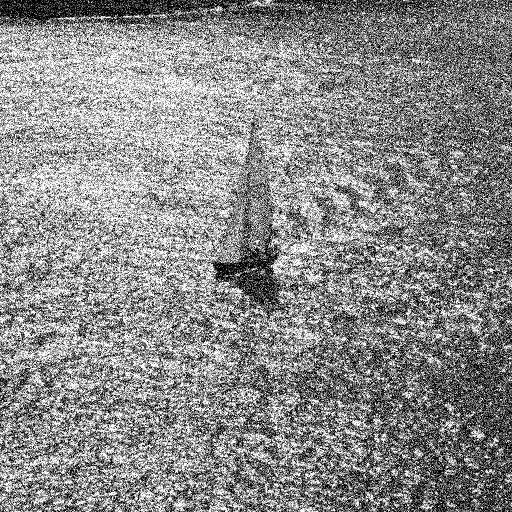}
	    \caption{A highly textured image (6184.pgm) from \cite{BOWS} with added zero mean Gaussian noise (variance $0.05$).}
	\end{subfigure} \\
	\begin{subfigure}{\columnwidth}\centering
	\includegraphics[width=0.85\linewidth, height=0.4\linewidth]{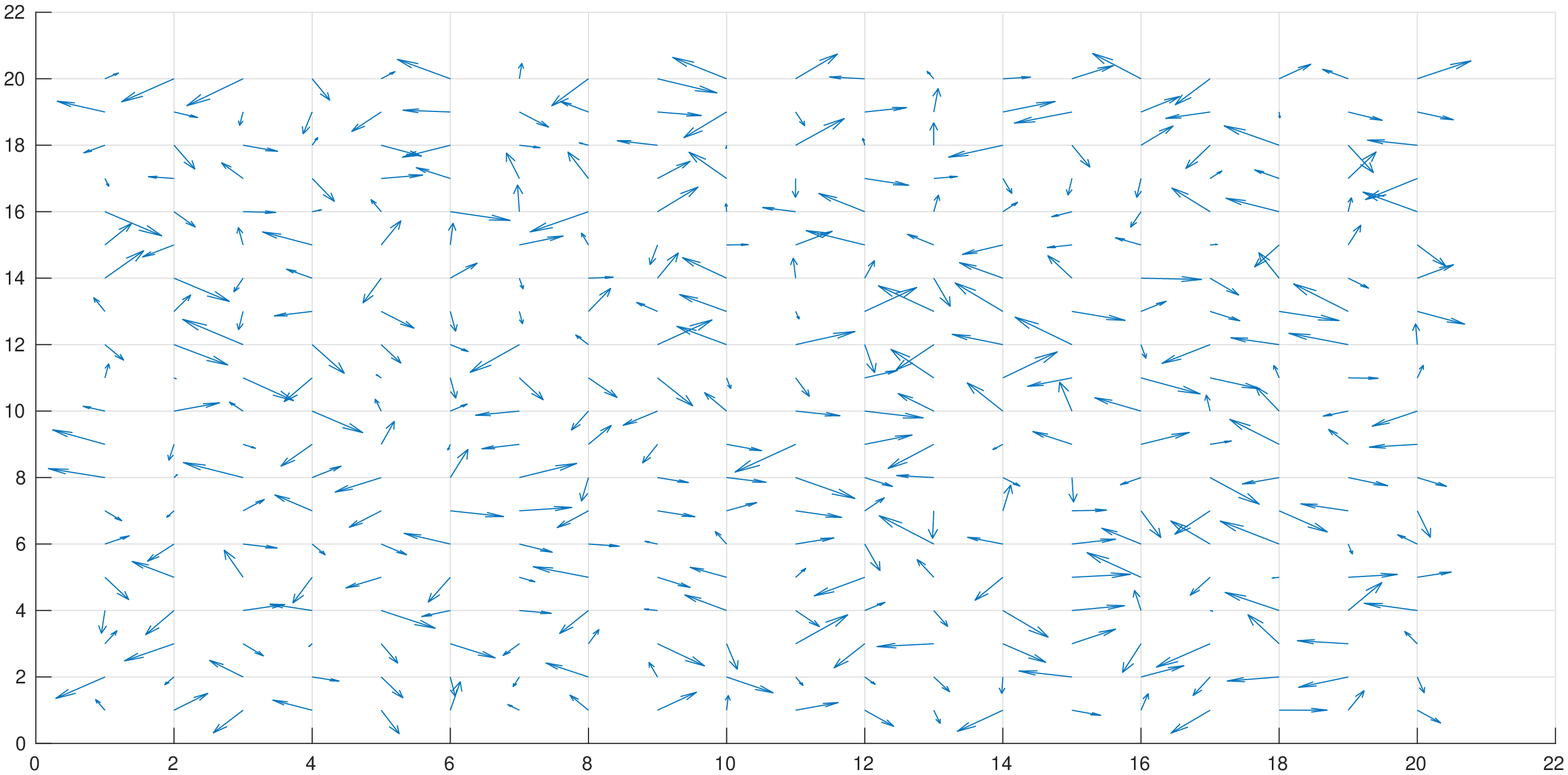}
	    \caption{A non-correlated motion field.}
	    \end{subfigure}
\caption{Examples of covers allowing high embedding rates.}
\label{fig:ex.high}
\end{figure}

\begin{figure}
\centering
	\begin{subfigure}{0.8\columnwidth}\centering
	\includegraphics[width=0.8\linewidth, height=0.4\linewidth]{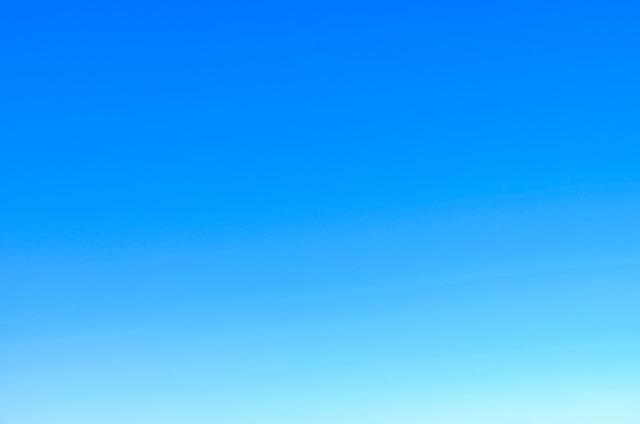}
	    \caption{A smooth image.}
	\end{subfigure} \\
	\begin{subfigure}{\columnwidth}\centering
	\includegraphics[width=0.8\linewidth, height=0.4\linewidth]{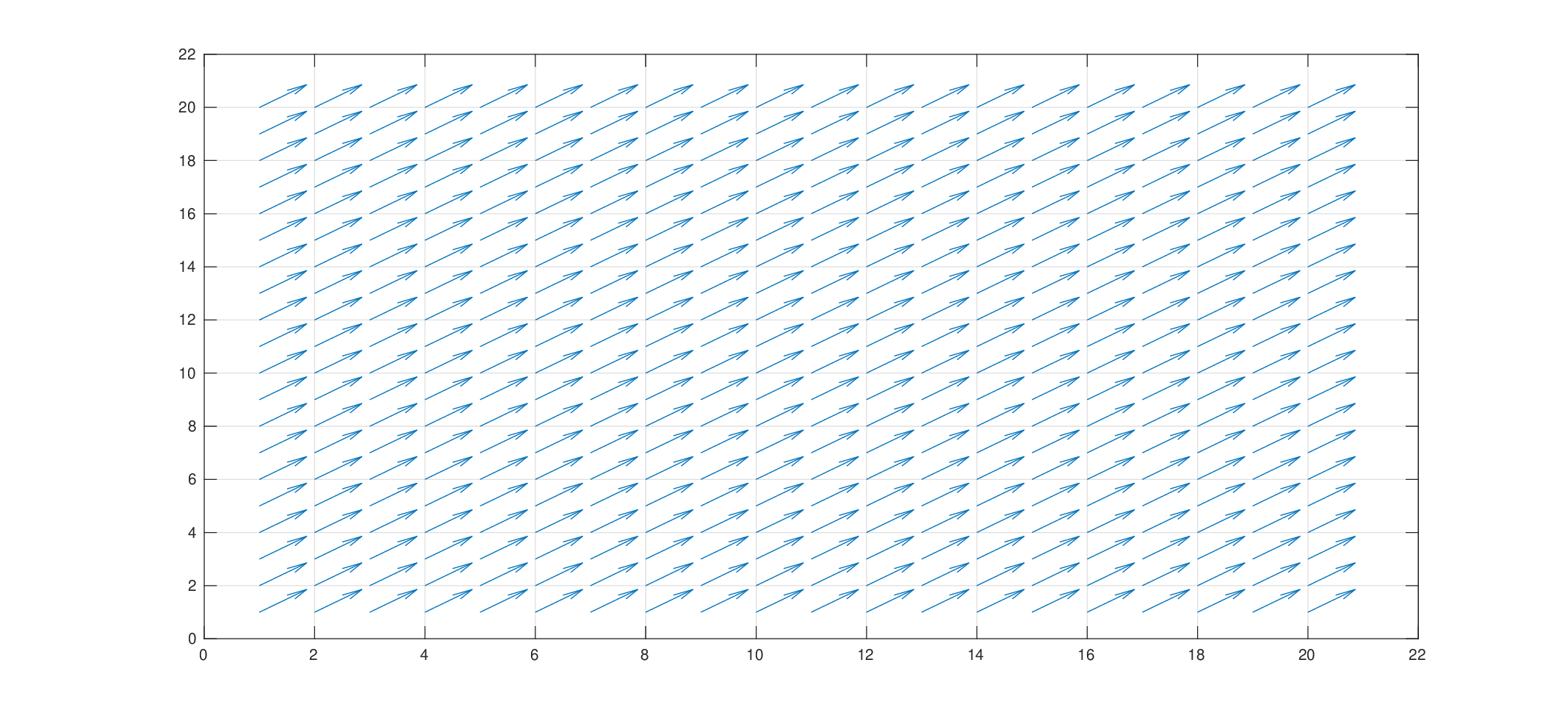}
	    \caption{A highly correlated motion field.}
	    \end{subfigure}
\caption{Examples of covers allowing low embedding rates.}
\label{fig:ex.low}
\end{figure}

\section{Conclusion}\label{sec:conc}
An information-theoretic upper bound was presented for data embedding with multimedia covers for a given level of security.
The contributions of this paper are as follows.
First, a justification was given for using a CMQGD to approximate the intractable Gibbs distribution.
Second, an analytic form for the SRL was derived.
Third, an analytic form was given not only for the case when the cover is known by both the sender $\Sender$ and receiver $\Receiver$,
but also for the practical case when $\Receiver$ only knows the cover distribution.
This was obtained using the achievability proof in Section \ref{sec:ach}.
Fourth, the bound was used with several types of steganographic detectors and multimedia.
Fifth, model parameters were given for the optimal message codebook in an analytic form.
Note that the results presented for cover objects with a continuous distribution are exact closed-form solutions.

\section*{Appendix A\\ Extension of Theorem 2 in \cite{gaussisgauss} for Non-zero Mean Gaussian Distributions}\label{sec.sub.prof}
The proof of Theorem 2 in \cite{gaussisgauss} for zero mean Gaussian distributions can be extended to the case of non-zero mean Gaussian distributions,
i.e. $p_0(y) \sim \mathcal{N}(\mu,\sigma_w)$.
Following the approach in \cite{gaussisgauss} except that $p_0(y)=\frac{1}{\sigma_w\sqrt{2\pi}} e^{\frac{-1}{2\sigma_w^2}(y-\mu)^2}$, we have
\begin{align}
\int_{-\infty}^{\infty} p_1(y)~dy &= \frac{e^{-\rho_0-1}}{\sigma_w\sqrt{2\pi}} \int_{-\infty}^{\infty} e^{\frac{-1}{2\sigma_w^2}(y-\mu)^2} e^{-\rho_1y^2} \quad dy, \nonumber \\
&= \frac{e^{-\rho_0-1-\frac{\mu}{2\sigma_w^2}}}{\sigma_w\sqrt{2\pi}} \int_{-\infty}^{\infty} e^{-y^2(\rho_1 + \frac{1}{2\sigma_w^2})} e^{\frac{\mu}{\sigma_w^2}y} \quad dy. \nonumber
\end{align}
Then using \cite[(3.323.2)]{mathbook.7.3.4} gives
\begin{equation*}
\int_{-\infty}^{\infty} e^{-a^2x^2\pm bx} dx = \frac{\sqrt{\pi}}{a}e^{\frac{b^2}{4a^2}},
\end{equation*}
so
\begin{align}\label{new.23}
\int_{-\infty}^{\infty} p_1(y)~dy &= \frac{1}{\sqrt{2\rho_1\sigma_w^2 +1}} e^{\rho_0 -1 -\frac{\mu}{2\sigma_w^2} + \frac{\mu^2}{4\rho\sigma^4_w+2\sigma_w^2} }. \\
\end{align}
Equation \eqref{new.23} is an extension of (23) in \cite{gaussisgauss}.
Similarly, for (26) in \cite{gaussisgauss} we have
\begin{align}
\int_{-\infty}^{\infty} y^2p_1(y)~dy &= \frac{e^{-\rho_0-1}}{\sigma_w\sqrt{2\pi}} \int_{-\infty}^{\infty} y^2 e^{\frac{-1}{2\sigma_w^2}(y-\mu)^2} e^{-\rho_1y^2} \quad dy, \nonumber \\
&= \frac{e^{-\rho_0-1-\frac{\mu}{2\sigma_w^2}}}{\sigma_w\sqrt{2\pi}} \int_{-\infty}^{\infty} y^2  e^{-y^2(\rho_1 + \frac{1}{2\sigma_w^2})}\nonumber \\
 & e^{\frac{\mu}{\sigma_w^2}y} \quad dy. \nonumber
\end{align}
Then from \cite[(3.462.8)]{mathbook.7.3.4}
\begin{equation*}
\int_{-\infty}^{\infty} x^2 e^{-ax^2} e^{2bx} dx = \frac{1}{2a} \sqrt{\frac{\pi}{a}} \left(1+ 2\frac{b^2}{a}\right) e^{\frac{b^2}{a}},
\end{equation*}
which gives
\begin{align}\label{new.26}
\int_{-\infty}^{\infty} y^2p_1(y)~dy &= \frac{e^{-\rho_0 -1 - \frac{\mu}{2\sigma_w^2}}\sigma_w^2}{(1+2\rho_1\sigma_w^2)^\frac{3}{2}} \left(1+ \frac{\mu}{4\sigma_w^4(\rho_1+\frac{1}{2\sigma_w^2})}\right) \nonumber \\
& e^{\frac{\mu^2}{4\sigma_w^4(\rho_1+\frac{1}{2\sigma_w^2})}}.
\end{align}
Substituting (17-b) from \cite{gaussisgauss} in \eqref{new.23} we have
\begin{equation}\label{new.25}
e^{\rho_0 -1} = \frac{1}{\sqrt{2\rho_1\sigma_w^2 +1}} e^{-\frac{\mu}{2\sigma_w^2} + \frac{\mu^2}{4\rho\sigma^4_w+2\sigma_w^2} },
\end{equation}
which is a generalization of (25) at \cite{gaussisgauss}.
In addition, substituting (17-c) from \cite{gaussisgauss} in \eqref{new.23} and \eqref{new.26} results in
\begin{equation*}
P_y=\frac{2\sigma^2_w (1+2\rho_1\sigma_w^2) +\mu}{2(2\rho_1\sigma_w^2+1)^2},
\end{equation*}
and solving for $\rho_1$ gives
\begin{equation}\label{new.28}
\rho_1 = \frac{\sigma_w^2 - 2P_y+\sqrt{\sigma_w^4+2\mu P_y}}{4\sigma_w^2 P_y},
\end{equation}
which generalizes (28) in \cite{gaussisgauss}.
Now, substituting \eqref{new.25} and \eqref{new.28} in (22) in \cite{gaussisgauss}, we have
\begin{equation}\label{new.29}
p_1(y) = \frac{1}{\sqrt{2\pi P_y}} e^{\frac{-1}{2P_y^2}(y-\mu)^2}.
\end{equation}
which generalizes (29) in \cite{gaussisgauss}.
Using this expression in the remainder of the proof of Theorem 2 in \cite{gaussisgauss} proves the generalization for non-zero mean Gaussian distributions,
and this can be extended to multivariate Gaussian distributions such as CMQGD. \hfill $\blacksquare$

\section*{Appendix B\\ Entropy of a Multivariate Quantized Gaussian Distribution}\label{lemma.1}

\begin{definition}\label{lem.A1}
The entropy of $\CoverDist$ and $\StegDist$.
\end{definition}
Theorem 8.3.1 in \cite{info.th.book} defines the entropy of a quantized distribution as
\begin{equation}\label{eq.ent.t.cover}
H(P) \approx h(p) + b,
\end{equation}
where $p$ is a continuous distribution, $P$ is a quantized version of $p$,
$b$ is the number of quantization bits, $H(P)$ is the entropy of the quantized distribution and $h(p)$ is the differential entropy of the continuous distribution.
The entropy of CMQGD $p_g$ is defined in \cite{limits2} as
\begin{equation}\label{eq.ent.mqgd}
h(p_g) = \frac{1}{2}\ln (2\pi e |\Sigma_g|),
\end{equation}
so the entropy of $\CoverDist$ and $\StegDist$ can be approximated w.r.t. the corresponding continuous versions of $\CoverDist$ and $\StegDist$ as
\begin{equation}\label{eq.ent.pc}
H(\CoverDist) \approx \frac{1}{2}\ln (2\pi e |\Sigma_c|) + b,
\end{equation}
\begin{equation}\label{eq.ent.ps}
H(\StegDist) \approx \frac{1}{2}\ln (2\pi e |\Sigma_s|) + b.
\end{equation}

\section*{Appendix C\\ Additional Properties of the Embedding Factor $a$}

\begin{cor}\label{lem.A*}
$a$ cannot be a complex number.
\end{cor}
Note that $a \in\mathbb{R}$ although $W(.) \in \mathbb{C}$.
From \cite{WrightOmega}, the relationship between the Lambert W Function $\mathcal{W}(.)$ and the WrightOmega function $W(.)$ can be expressed as
\begin{equation}\label{eq.lambertw}
W(x) = \mathcal{W}_{\big \lceil \frac{\mathrm{Im}(x) - \pi}{2 \pi} \big \rceil}(e^{x}).
\end{equation}
If $x = -\frac{4\epsilon^2}{n} -1 -i\pi$ as in \eqref{eq.result4m}, then $\mathrm{Im}(x) = \pi$ and therefore \eqref{eq.lambertw} can be rewritten as
\begin{equation}\label{eq.lambertw.2}
W(x) = \mathcal{W}_{-1}(e^{x}),
\end{equation}
which is a special case of the Lambert W Function where $\mathcal{W}_{-1}(.) \in \mathbb{R}$ \cite{LambertWfunction}.
Thus, $a* \in \mathbb{R}$, and this also applies to $a$ in \eqref{eq.result4}.

\begin{cor}\label{lemma.a.geq.1}
$a$ cannot be less than 1.
\end{cor}
As $\Sigma_s, \Sigma_c$ and $\Sigma_m$ are positive semi-definite matrices,
from \eqref{eq.sigmam2} $a$ must be greater than 1 with equality when no embedding occurs from \eqref{eq.result2.a} and \eqref{eq.sigmam2}.
Thus
\begin{equation}\label{eq.agt1}
a \geq 1.
\end{equation}

\begin{cor}\label{lamma.a.pe}
The relation between $\AvgProbErrorEavesdropper$ and $a$.
\end{cor}
From \eqref{eq.pe_pd}-\eqref{eq.klv} and \eqref{eq.c1}, we have
\begin{equation}\label{eq.dpe}
2\,(1-2\AvgProbErrorEavesdropper)^2 \leq \mathcal{D}(\StegDist \parallel \CoverDist) \leq 2\epsilon^2,
\end{equation}
Then, as the WrightOmega function in the form $-W(-\gamma-1-i\pi)$ is monotonically increasing with $\gamma$ \cite{WrightOmega}, where $\gamma$ is a positive constant,
from  \eqref{eq.result4}, \eqref{eq.result4m} and \eqref{eq.dpe} we have that
\begin{equation}\label{eq.result4.upper}
-W(-\frac{4\,(1-2\AvgProbErrorEavesdropper)^2}{n} -1 -i\pi) \leq a \leq -W(-\frac{4\epsilon^2}{n} -1 -i\pi).
\end{equation}

\bibliographystyle{IEEEtran}
\bibliography{IEEEabrv,main}

\clearpage
\twocolumn[\section*{\Large Supplementary Document for: ``Information-Theoretic Bounds for Steganography in Multimedia"}]

\begin{figure}[ht]\centering
	\includegraphics[width=0.8\linewidth]{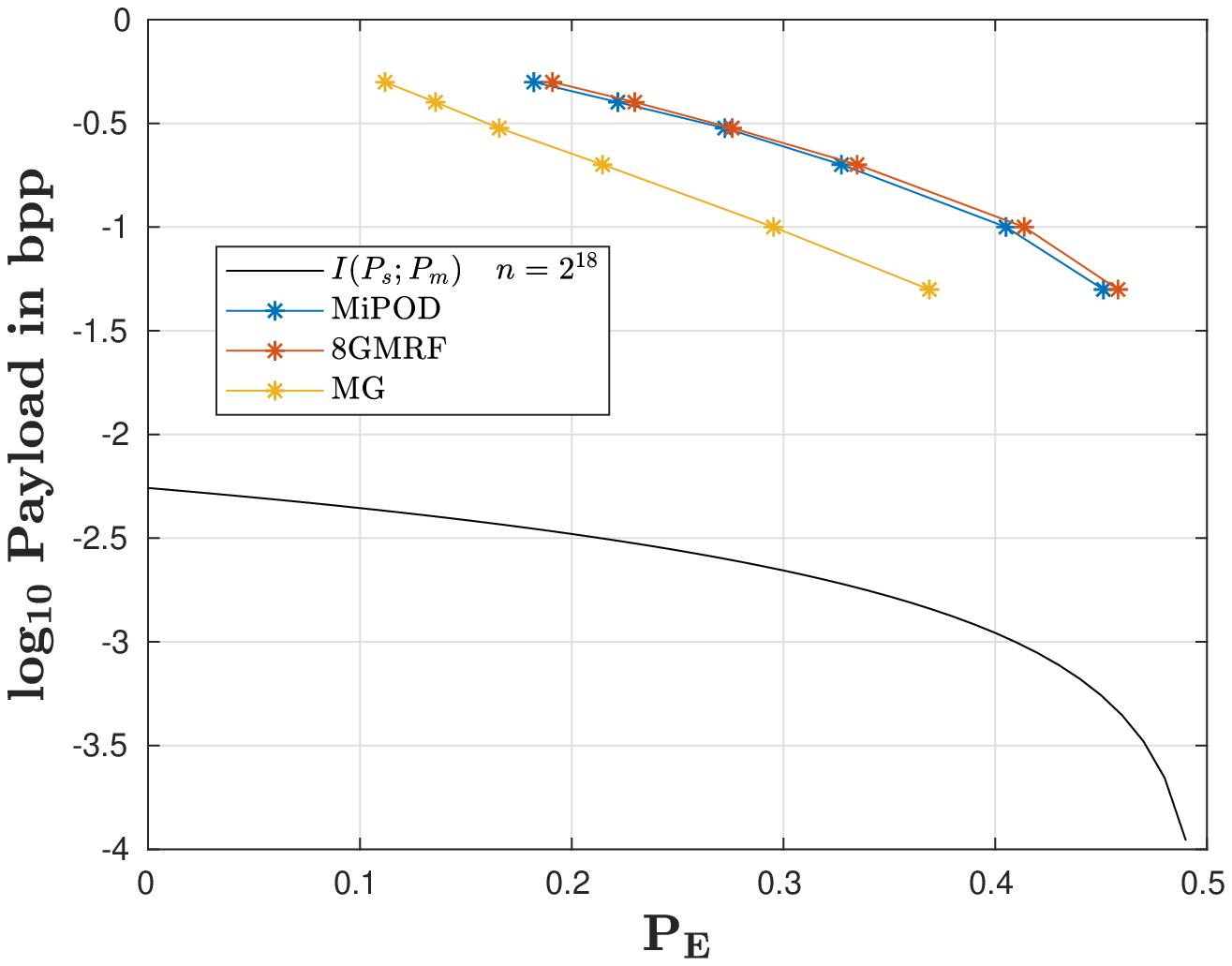}
	    \caption{Results from \cite{markov8} for the MiPOD, 8-GMRF, and MG steganographic methods with steganalysis using SRM.}\label{fig:Markov_8_T1a}
\end{figure}
\begin{figure}\centering
	\includegraphics[width=0.8\linewidth]{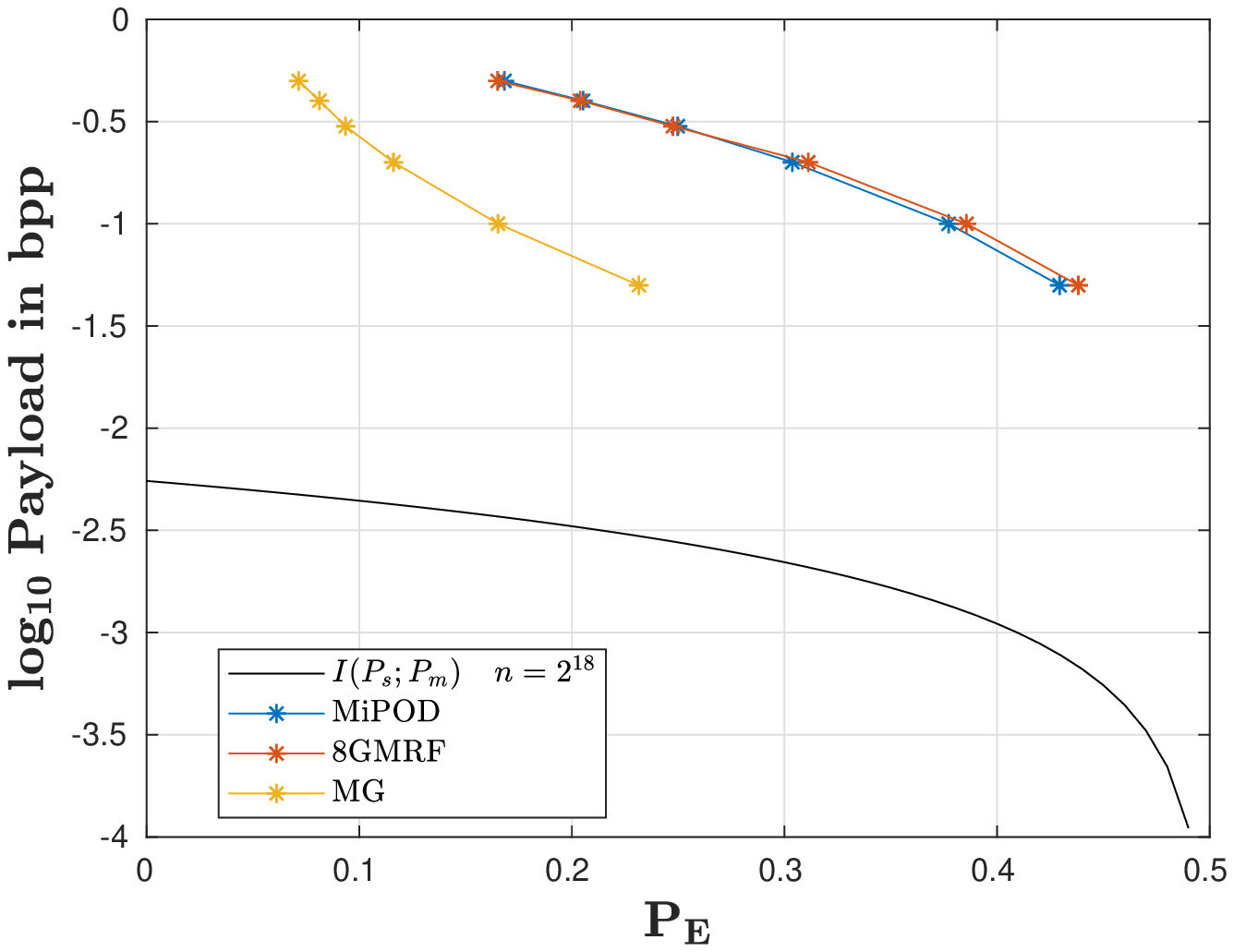}
	    \caption{Results from \cite{markov8} for the MiPOD, 8-GMRF, and MG steganographic methods with steganalysis using maxSRMd2.}\label{fig:Markov_8_T1b}
\end{figure}
\begin{figure}\centering
	\includegraphics[width=0.8\linewidth]{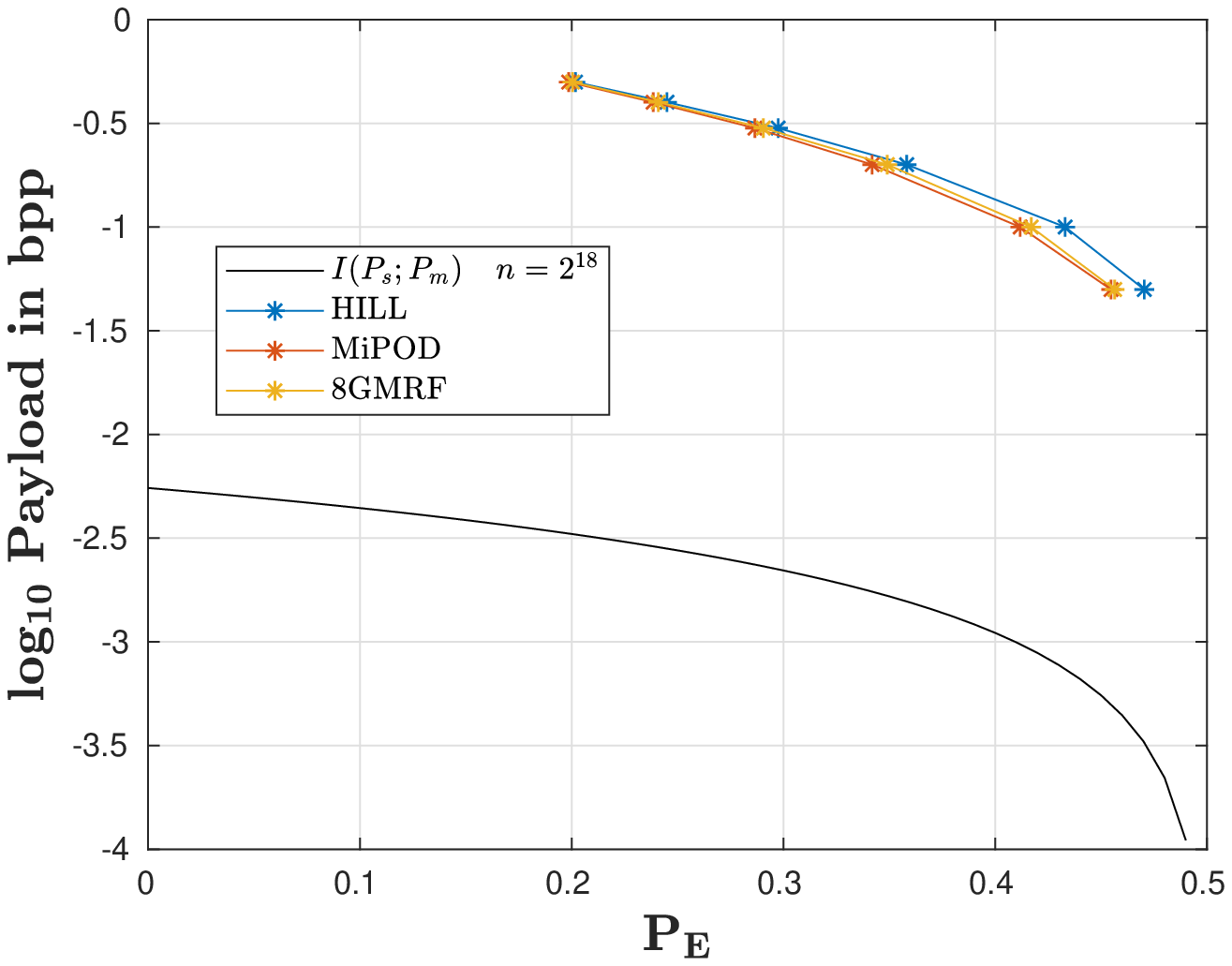}
	    \caption{Results from \cite{markov8} for the MiPOD and 8-GMRF (both enhanced using low-pass filtered cost) and HILL steganographic methods and the ensemble 1.0 classifier with steganalysis using SRM.}\label{fig:Markov_8_T2a}
\end{figure}
\begin{figure}\centering
	\includegraphics[width=0.8\linewidth]{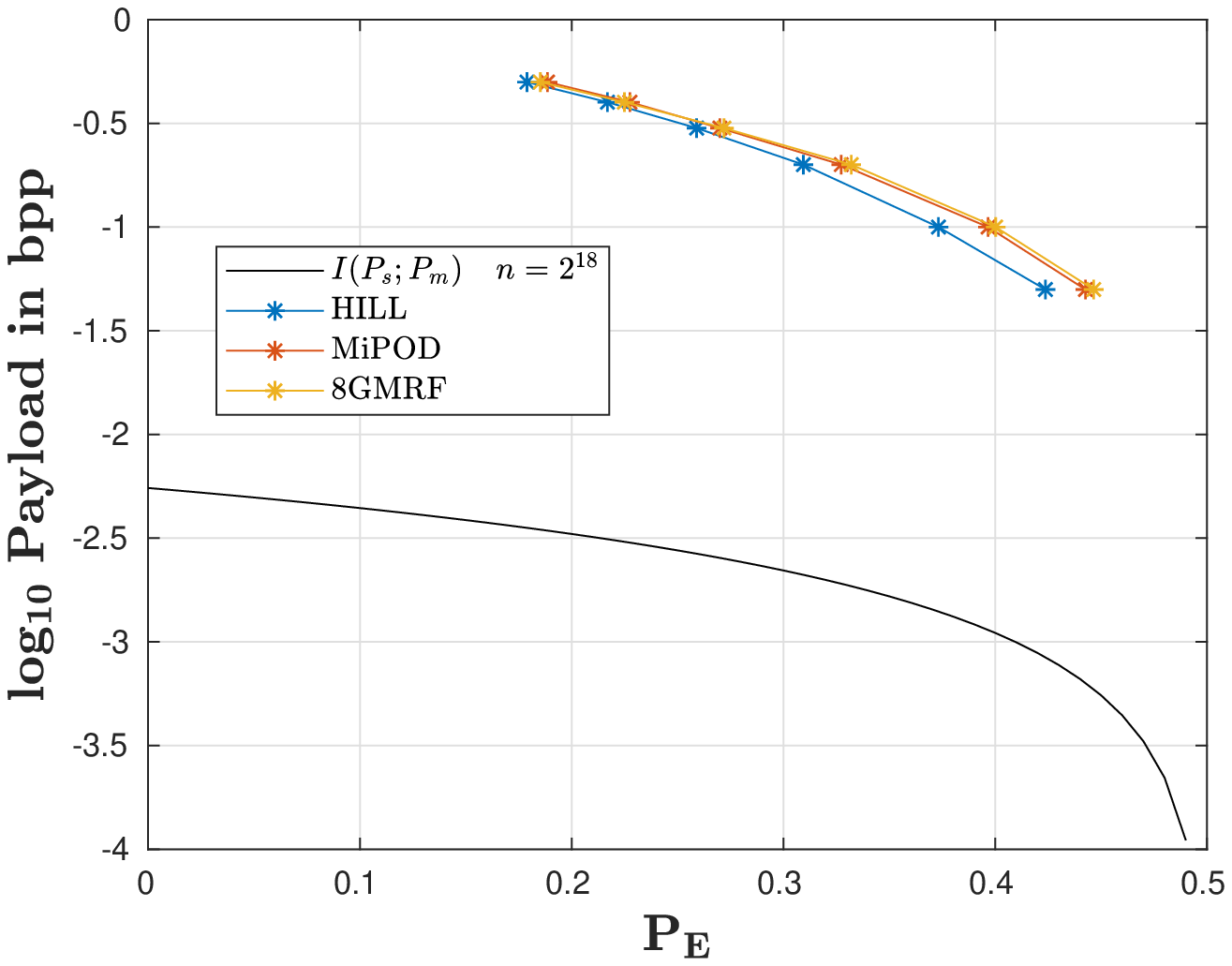}
	    \caption{Results from \cite{markov8} for the MiPOD and 8-GMRF (both enhanced using low-pass filtered cost) and HILL steganographic methods and the ensemble 1.0 classifier with steganalysis using maxSRMd2.}\label{fig:Markov_8_T2b}
\end{figure}
\begin{figure}
  \centerline{\includegraphics[width=0.8\linewidth]{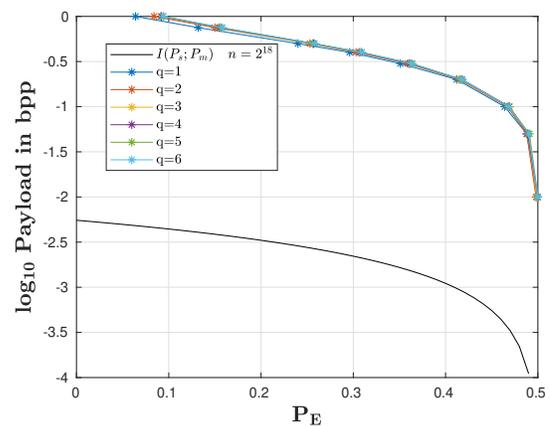}}
  \caption{Results for the Gaussian version of the HILL algorithm in \cite{quantized.gauss} with different $q$ values in a $(2q+1)$-ary embedding with steganalysis using maxSRMd2.}
  \label{fig:quantized_gauss_T1}
\end{figure}
\begin{figure}\centering
\includegraphics[width=0.8\linewidth]{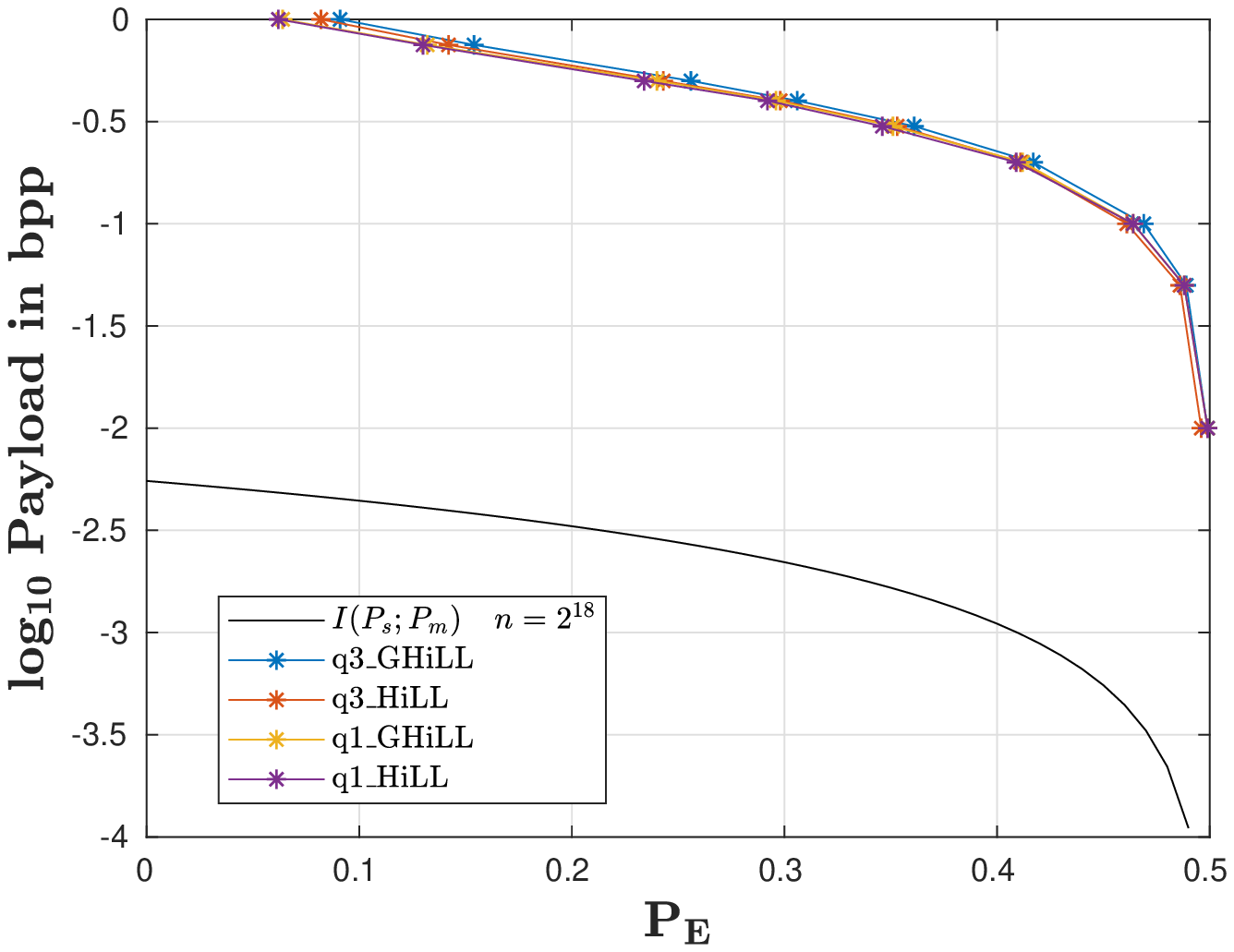}
    \caption{Results from \cite{quantized.gauss} for the HILL steganographic method and the modified Gaussian version for $q=1,3$ with steganalysis using maxSRMd2.}
    \label{fig:quantized_gauss_T2_HILL}
\end{figure}
\begin{figure}\centering
\includegraphics[width=0.8\linewidth]{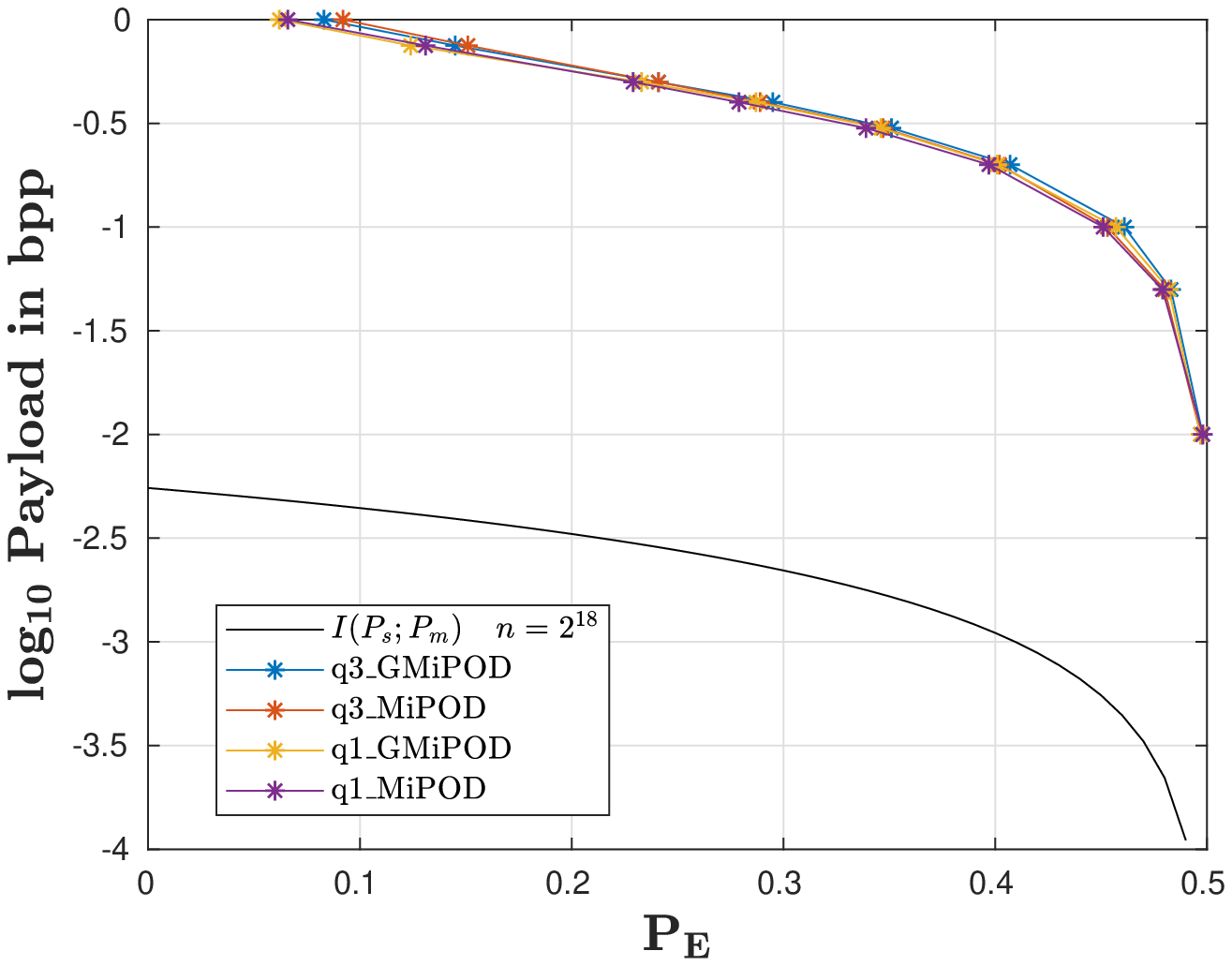}
    \caption{Results from \cite{quantized.gauss} for the MiPOD steganographic method and the modified Gaussian version for $q=1,3$ with steganalysis using maxSRMd2.}
    \label{fig:quantized_gauss_T2_MiPOD}
    \end{figure}
\begin{figure}\centering
\includegraphics[width=0.8\linewidth]{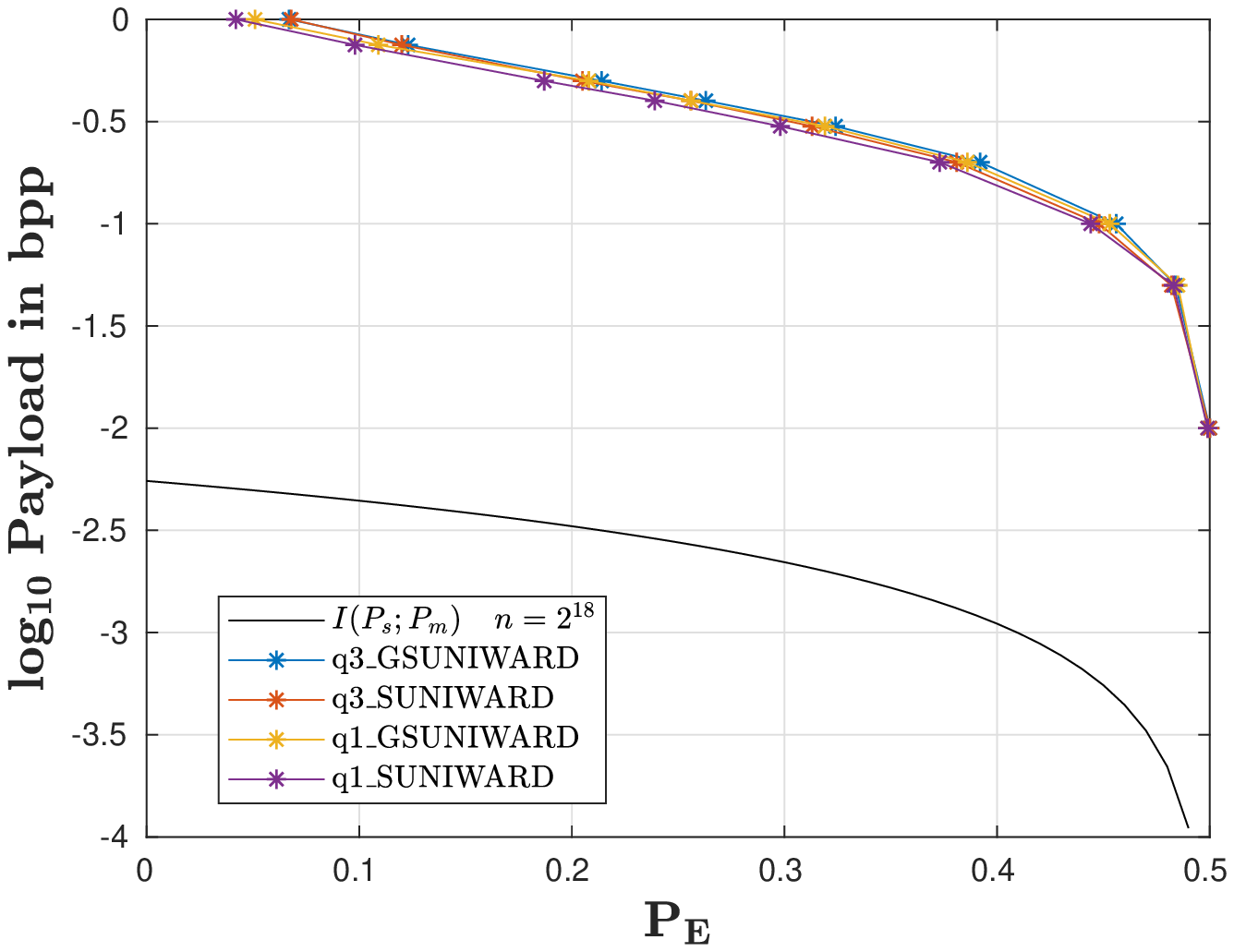}
    \caption{Results from \cite{quantized.gauss} for the SUNIWARD steganographic method and the modified Gaussian version for $q=1,3$ with steganalysis using maxSRMd2.}
    \label{fig:quantized_gauss_T2_SUNIWARD}
\end{figure}
\begin{figure}\centering
\includegraphics[width=0.8\linewidth]{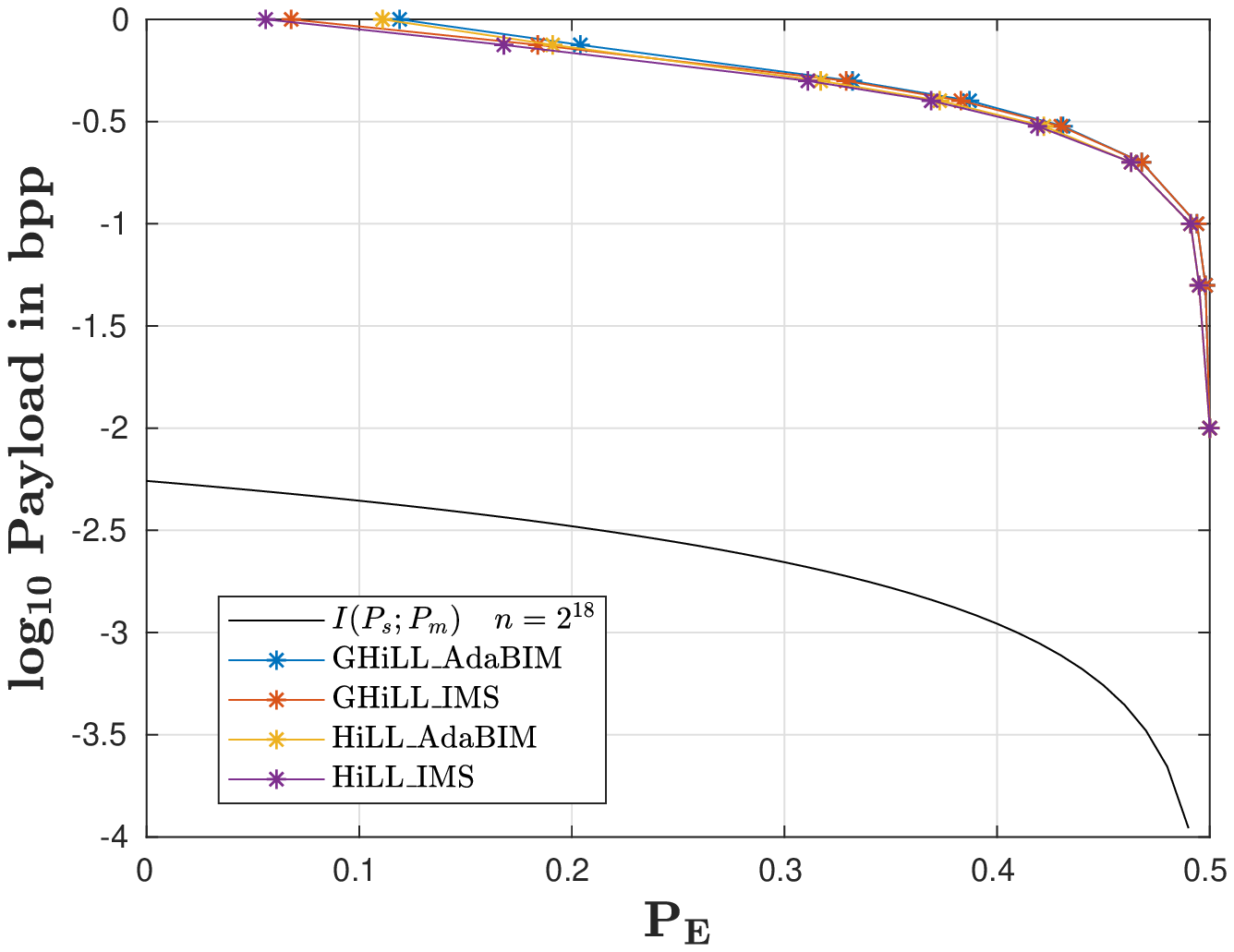}
    \caption{Results from \cite{quantized.gauss} for the HILL steganographic method and the modified Gaussian version with two batching strategies, IMS with batch size 128 and AdaBIM with adaptive batch size, with steganalysis using maxSRMd2.}
    \label{fig:quantized_gauss_T3_HILL}
\end{figure}
\begin{figure}\centering
\includegraphics[width=0.8\linewidth]{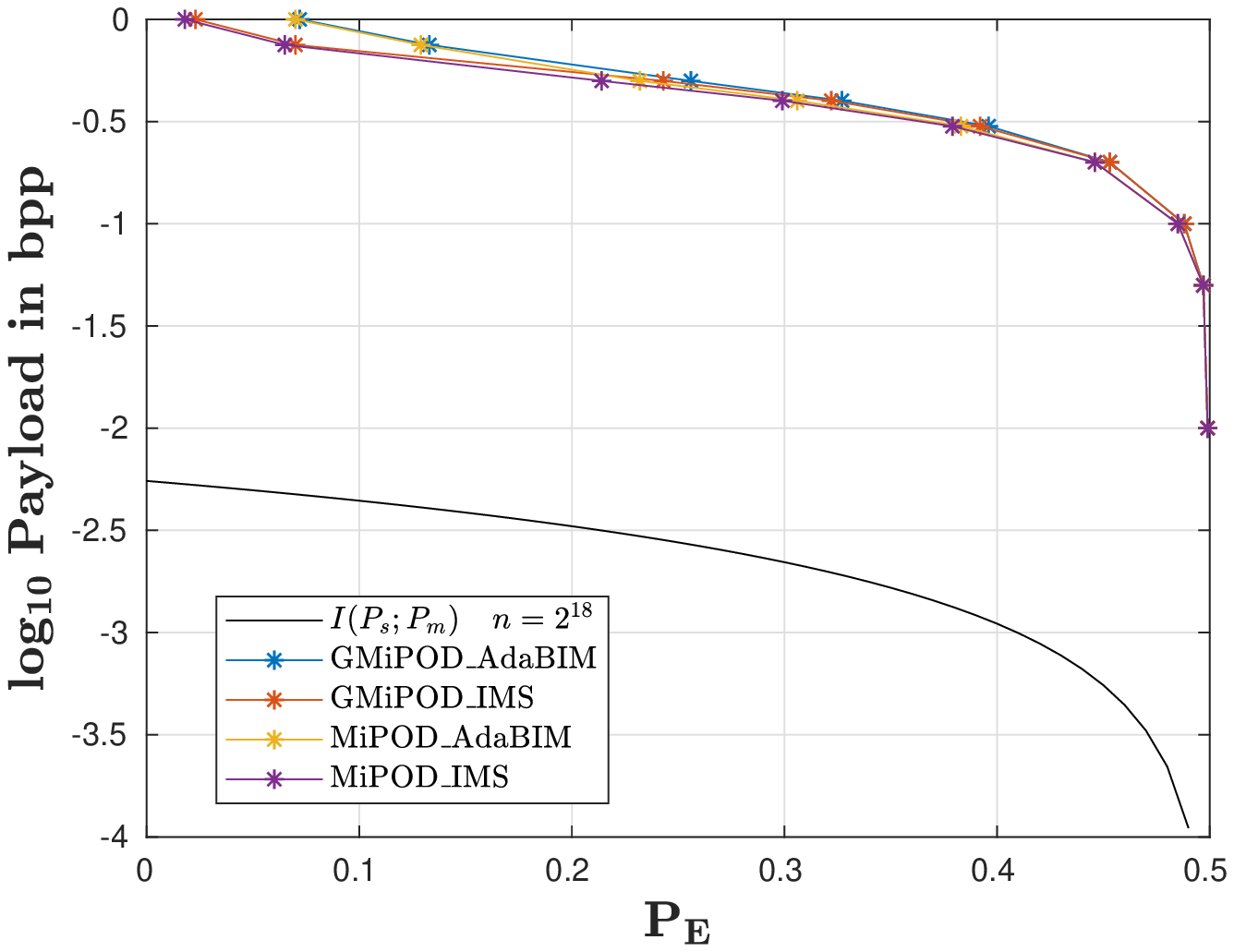}
    \caption{Results from \cite{quantized.gauss} for the MiPOD steganographic method and the modified Gaussian version with two batching strategies, IMS with batch size 128 and AdaBIM with adaptive batch size, with steganalysis using maxSRMd2.}
    \label{fig:quantized_gauss_T3_MiPOD}
    \end{figure}
\begin{figure}\centering
\includegraphics[width=0.8\linewidth]{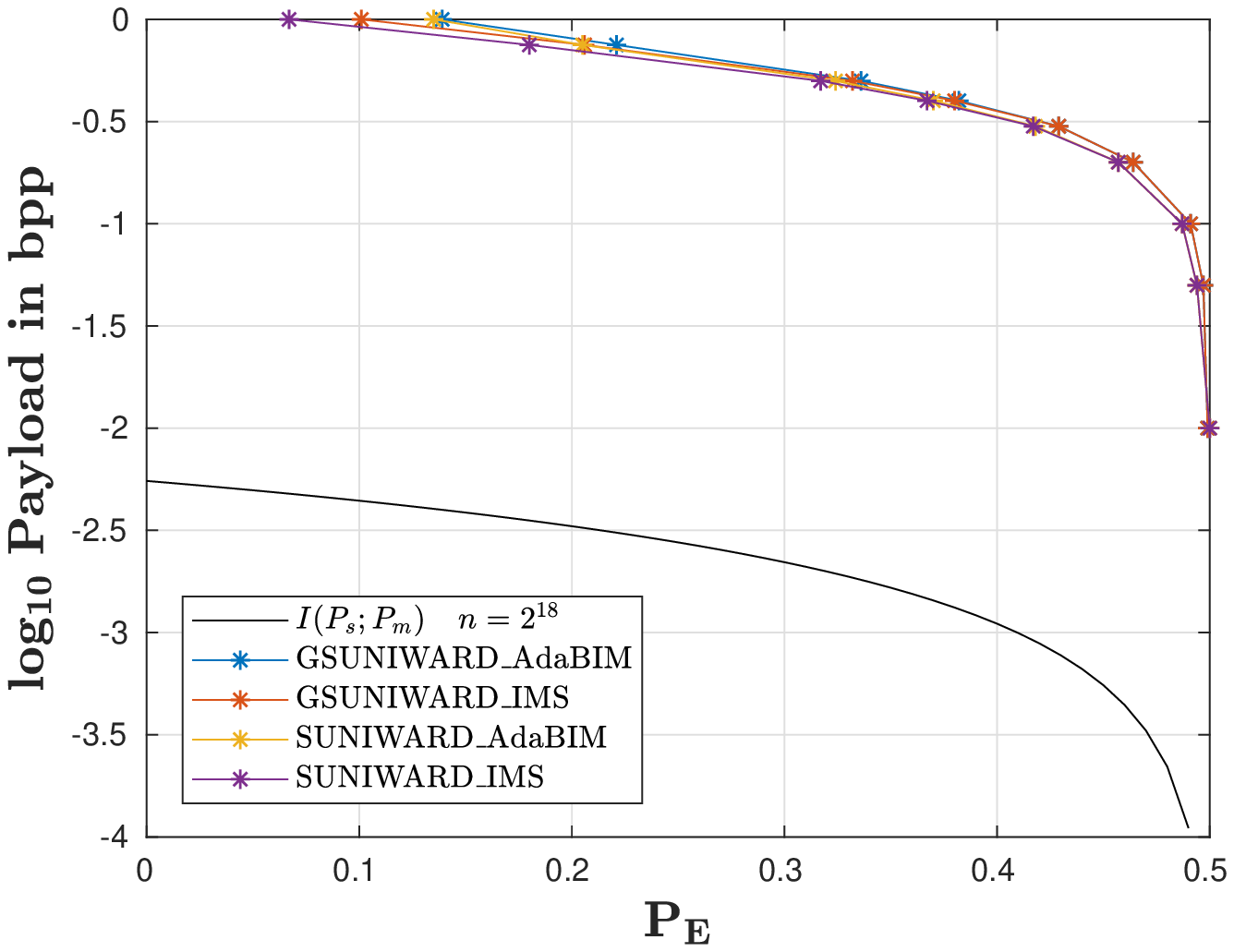}
    \caption{Results from \cite{quantized.gauss} for the SUNIWARD steganographic method and the modified Gaussian version with two batching strategies, IMS with batch size 128 and AdaBIM with adaptive batch size, with steganalysis using maxSRMd2.}
    \label{fig:quantized_gauss_T3_SUNIWARD}
\end{figure}
\end{document}